\documentclass[11pt,reqno]{amsart}
\pdfoutput=1 
\usepackage{xypic}
\usepackage{graphicx}
\usepackage{amsfonts}
\usepackage{mathrsfs}
\usepackage{amssymb}
\input{epsf.sty}
\catcode `\@=11
\def\numberbysection{\@addtoreset{equation}{section}
         \renewcommand{\theequation}{\thesection.\arabic{equation}}}
\numberbysection
\def\subsubsection{\@startsection{subsubsection}{3}%
  \normalparindent{.5\linespacing\@plus.7\linespacing}{-.5em}%
  {\normalfont\bfseries}}
\newtheorem{thm}{Theorem}[section]
\newtheorem{lem}[thm]{Lemma}
\newtheorem{prop}[thm]{Proposition}
\newtheorem{cor}[thm]{Corollary}

\newtheorem{ep}[thm]{Expectation}

\theoremstyle{definition}

\newtheorem{df}[thm]{Definition}

\newtheorem{rmk}[thm]{Remark}

\newtheorem{ex}[thm]{Example}

\def\nn{\nonumber}

\def\eps{\epsilon}
\def\G{\Gamma}

\def\R{{\mathbb R}}
\def\C{{\mathbb C}}
\def\Z{{\mathbb Z}}
\def\T{{\mathbb T}}
\def\F{{\mathbb F}}
\def\H{{\mathscr H}}

\def\Lhex{\Lambda_{hex}}
\def\Lhc{\Lambda_{hc}}
\def\B{{\mathscr B}}
\def\L{\Lambda}

\def\formTheta{\widehat\Theta}
\def\BformTheta{\B_{\formTheta}}
\def\BTheta{\B_{\Theta}}

\def\grep{\rho_{\Gamma}}

\def\a{\alpha}

\def\bgl{\bar\G_{\L}}
\def\gl{\G_{\L}}
\def\gp{\Gamma_+}
\def\bgp{\bar\Gamma_+}
\def\Cliff{\it Cliff}
\def\cube{\bar\Gamma_+^{\it crystal}}
\def\Lhon{\L_{\rm hon}}

\begin{document}

\title[The geometry of the Double Gyroid wire network]{
The geometry of the Double Gyroid wire network: Quantum and Classical}

\author
[Ralph M.\ Kaufmann]{Ralph M.\ Kaufmann}
\email{rkaufman@math.purdue.edu}

\address{Department of Mathematics, Purdue University, 
 West Lafayette, IN 47907
 and School of Mathematics, Institute for Advanced Study, Princeton NJ, 08540}

\author
[Sergei Khlebnikov]{Sergei Khlebnikov}
\email{skhleb@physics.purdue.edu}

\address{Department of  Physics, Purdue University,
 West Lafayette, IN 47907}

\author
[Birgit Kaufmann]{Birgit Wehefritz--Kaufmann}
\email{ebkaufma@math.purdue.edu}

\address{Department of Mathematics and Department of Physics, Purdue University,
 West Lafayette, IN 47907
 and  Department of Physics, Princeton University, Princeton NJ, 08544}

\subjclass[2000]{46L60, 81R60, 58B34, 53A10, 57M15, 52C07, 05C38, 81Q10}

\begin{abstract}
Quantum wire networks have recently become of great interest. Here
we deal with a novel nano material structure of a Double Gyroid wire
network. We use methods of commutative and 
non-commutative geometry to describe this
wire network.
Its non--commutative geometry is closely related to  non-commutative 3-tori as we
discuss in detail.
\end{abstract}

\maketitle

\section*{Introduction}
Interfaces that can be modeled by surfaces of constant mean curvature (CMC) are
ubiquitous in nature and can now be synthesized in laboratory. Recently, Urade et al.\
\cite{Hillhouse} have reported fabrication of a nano-porous silica film whose structure
is related to a specific CMC surface---the Gyroid. The structure
is three-dimensionally periodic and has three components:\footnote{And so
is sometimes referred to as tri-continuous.}
a thick surface and two channels, as detailed below. The interface between the wall
and the channels approximates a Double Gyroid.

Urade et al.\ \cite{Hillhouse} 
have also demonstrated a nanofabrication technique in which the channels are 
filled with a metal, while the silica wall can be either left in place or removed.
These novel materials open a wide field
of applications due to their topological and geometric features. The channels
are a few nanometers wide and, when filled with a conducting or semiconducting
material, are expected to acquire certain characteristics of one-dimensional 
quantum wires (such as a blueshift of the spectrum and an enhanced density of 
states), while remaining three-dimensional in other respects \cite{KW}.
Geometrically, the one--dimensional structure appears since
each of the channels can be retracted to a skeletal graph \cite{KGB}, 
the Gyroid graph.

We will concentrate on the resulting geometry and topology of these
networks in the possible presence of a constant magnetic field.
Topological quantities are of particular physical interest as they
remain stable under continuous perturbations. In practice,
those should be small as not to break the structure.

We use two approaches:
the first is purely classical and the
second is the quantum/non--commutative approach due to Alain Connes \cite{Connes}.
The classical results we provide are a study of the fundamental group
of the channel. Here we determine the fundamental group of each channel
system. This group is the commutator subgroup of the free group in
three generators. A surprising result coming directly from
the Gyroid geometry is that there is a new length function on
the free group which is different from the ordinary word length.
With the help of this length function we determine the shortest
loops at any given point. There are 15 such topologically distinct
shortest loops (30 if one includes orientation). They split into two groups
which are distinguished by their cyclic symmetry which is either of order two
or three. In both groups there are three generators using this additional
symmetry. These loops are of particular 
interest since numerical simulations \cite{KW} show the possibility of
an enhanced density of states in a double-gyroid quantum wire, due to states that
are nearly localized near such loops. We also calculate
the flux of a constant magnetic field
through these loops. The tool is an effective unit vector.
The result of the calculation is that these effective
unit vectors have a particularly simple form (see Table \ref{looptable}).
This fact should also be relevant for the study of the spin--orbit coupling
of the loop--localized states.

Our study of the non--commutative geometry is motivated by 
one of the big early successes of the
non-commutative geometry of Alain Connes. This was the description by
Bellissard et al. of the quantum Hall effect \cite{BE}. This allowed to explain
the integer effect in terms of the non-commutative geometric
properties. The underlying geometry in that situation is the quantum
2-torus. Recently there have been further analyses on the fractional
effect using hyperbolic geometry as a model \cite{Marcolli}.
The conceptual approach as outlined in \cite{belissard} is to replace
the Brillouin zone by a non--commutative Brillouin zone which is given
by a $C^*$--algebra that contains the translational symmetry operators 
and the Hamiltonian. In this geometry relevant quantities can often be 
expressed in terms of the $K$ theory of this algebra. This Abelian group 
captures information related to the the topology or better the homotopy type
of the algebra or geometric setup. These are again quantities which are
stable under continuous deformations. A prime example is the Hall conductance.
 
A general fact which is pertinent to our 
discussion is that the $K$--groups also serve to label the gaps 
in the spectrum. 
Roughly this goes as follows. Given a gap in the spectrum there is a projector
projecting to the energy levels below the gap. This projector
in turn gives rise to a $K$--theory element. If one knows the ordered
$K$--theory then one can also sometimes deduce if only 
 finitely or infinitely many gaps are possible.

In our situation, we determine the said $C^*$ algebra for 
one channel in the presence of a magnetic field and describe
the $K$--theories. The Hamiltonian we use is the generalization  
\cite{Sunada,belissard} of the Harper Hamiltonian \cite{Harper,belissard} 
adapted to our situation\footnote{The generalized Harper operator is also the operator
underlying the non--commutative geometry of the quantum Hall effect \cite{BE,Marcolli}}.   
We call the resulting algebra the Bellissard--Harper algebra and
denote it by $\B$. This geometry is closely related to 
the non--commutative 3--torus 
$\T^3_{\Theta}$. Here $\Theta$ is a 
skew--symmetric $3\times 3$ matrix determined by the magnetic field.
In fact we show that the algebra is isomorphic
to a subalgebra of the $4\times 4$ matrix algebra with coefficients in
the non--commutative 3--torus. By varying the magnetic field,
we obtain a three--parameter family of algebras. We prove that at all but
{\em finitely many} points
this algebra is the full matrix algebra and hence Morita 
equivalent to $\T^3_{\Theta}$. Since $\T^3_{\Theta}$ is simple
at irrational $\Theta$ one would expect this to be true on a dense set (see \S 3). 
We not only prove that this expectation holds at the irrational points, but are 
able to extend this result to almost all rational points.
At  certain special values of the 
field which we enumerate, however, the algebra is genuinely smaller leading
to a possibly different $K$--theory. At these points the material 
may also exhibit special properties. 

 The ordered $K$--theory
of the 3--torus is completely known \cite{Rieffel,PV,E} and actually
completely classifies the isomorphism classes of such tori \cite{RS,EL}. Our algebra
always injects into a $\T^3_{\Theta}$. From this and the knowledge of 
the ordered $K$--group of the non--commutative torus, we obtain the result
that there are only finitely
many gaps possible at rational values of $\Theta$. 

Our approach to both the classical and the quantum geometry uses 
graph theoretical methods. 
The relevant graphs are the Gyroid graph which is a 3--regular graph 
and its quotient under the translation symmetry group.
In order to give the matrix algebras
explicitly we introduce the notions of a graph representation.
Using rooted spanning trees, we are able to represent
the algebra $\B$ in terms of matrices, where they are amenable 
to direct computations.

As the geometry of the Gyroid and its channels is quite difficult,
we also treat the honeycomb lattice as a two--dimensional analogue.
Indeed, the Gyroid lattice graph (the graph onto which each channel retracts)
is in many ways the three dimensional analogue
of the honeycomb lattice,  which is why when  developing the
 more general parts of the theory
we will parallely consider these two cases as our main examples. Both graphs
are 3--regular,
and both of the graphs are not mathematical lattices but 
only physical lattices\footnote{See \S\ref{maxsymsec} for the disambiguation
of the use of the word ``lattice''.}. This means that they give
 rise to two groups, one which is the space or symmetry group and the other 
is the group of the lattice of which 
 they can be considered a subset. In the honeycomb  case these are
 both $\Z^2$ but embedded into $\R^2$ as the triangular lattice and its dual.
For the Gyroid the groups are both $\Z^3$, but they are embedded as
 a body centered cubic (bcc) and  a face centered cubic  (fcc) lattice in $\R^3$ which are again dual to each other.
Furthermore  the fundamental group of the honeycomb lattice is the 
commutator subgroup of the free group on two elements, while for the 
Gyroid $\pi_1$ is the commutator subgroup  of the free group 
in three generators. And in both cases, we find a new length function
on the free group induced by the geometry of the lattice.

The paper is organized as follows:

In the first section, we 
start with a review of the classical geometry of the Gyroid and
then prove the results on the fundamental group and the smallest loops.
In the second section, we formally introduce graphs and lattices
and the relevant groups associated to them. We then 
define their representation in Hilbert spaces and the associated
Harper operators. Finally, we show 
how to obtain a matrix representation of $\B$ using rooted spanning trees.
In the third section, we apply the general theory of the second section to
the non--commutative geometry of a lattice using a Harper operator, 
detailing the honeycomb and the Gyroid case. Here we also briefly review
projective representations and explain how they  arise in
the presence of a magnetic field.
In this context, we can already show that there can only be finitely many gaps
for the Gyroid at rational $\Theta$. We also outline the general
approach to calculating $\B$ and its $K$--theory and the expected
results. In particular, we compute the $K$--theory in the commutative case
in terms of a cover of a torus.

In the last section we apply the outlined strategies to
 calculate the algebras $\B$ and their $K$--theory
for Bravais lattices, the honeycomb lattice and the Gyroid lattice graph.

\section{The classical geometry of the double Gyroid}
\label{gclasssec}
\subsection{The double Gyroid and its channels}
The Gyroid is an embedded CMC surface in $\R^3$ \cite{KGBW}. It was 
discovered by Alan Schoen \cite{Sch}. In nature it was observed as an interface for di-block co-polymers \cite{copolyone,copolytwo}.
The interface actually consists of {\em two} disconnected surfaces. {\em Each} of them
is a Gyroid surface. The Double Gyroid (DG) is a particular configuration 
of two mutually non--intersecting embedded Gyroids.

A single Gyroid has symmetry group $I4_132$ while the double Gyroid has the symmetry group $Ia\bar3d$ where the extra symmetry comes from interchanging
the two Gyroids \footnote{Here $I4_132$ and $Ia\bar3d$ are given in the international or Hermann--Mauguin notation for symmetry groups, see e.g.\cite{symnota}. }.

Since CMC surfaces are mathematically hard to handle level surfaces have been
suggested as a possible approximation in \cite{lambert}. 
The level surface approximation for the Double Gyroid again consists of
two level surfaces. We will call the two surface interfaces $S_1$ and $S_2$.
For the discussion at hand it is not relevant if the two surfaces
are actually the CMC surfaces or their level surface approximations.
One example of a DG approximation is given by 
the family of  level surfaces \cite{lambert}

\begin{equation}
L_t:\sin x\cos y + \sin y \cos z+ \sin z \cos x=t
\end{equation}

A model for the double Gyroid is then given by 
$L_{w}$ and $L_{-w}$ for $0\leq w< \sqrt {2}$.

The complement $\R^3\setminus G$ 
of a single Gyroid $G$ has two components. These components
will be called the Gyroid wire systems or channels. 

There are two distinct channels, one left and one right handed.
Each of these channels contracts onto a graph, called skeletal 
graph in \cite{Sch,KGB}. We will call these graphs $\Gamma_+$ and $\Gamma_-$.
Each graph is periodic and trivalent. We fix $\Gamma_+$ to be the
graph which has the node $v_0=(\frac{5}{8},\frac{5}{8},\frac{5}{8})$.
We will give more details on the graph $\Gamma_+$ below.

The channel containing $\Gamma_+$ is shown in Figure \ref{bigcellfig}.
A (crystal) unit cell of the channel together with the embedded graph $\gp$
is shown in Figure \ref{gyrlabelsfig}
and just the graph is contained in Figure \ref{gammalabelsfig}.

\begin{figure}
\includegraphics{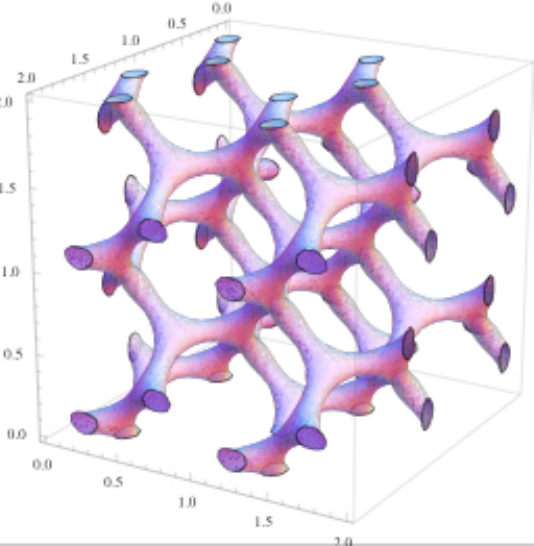}
\caption{\label{bigcellfig}The Channel $C_+$}
\end{figure}

\begin{figure}
\includegraphics[width=.7\textwidth]{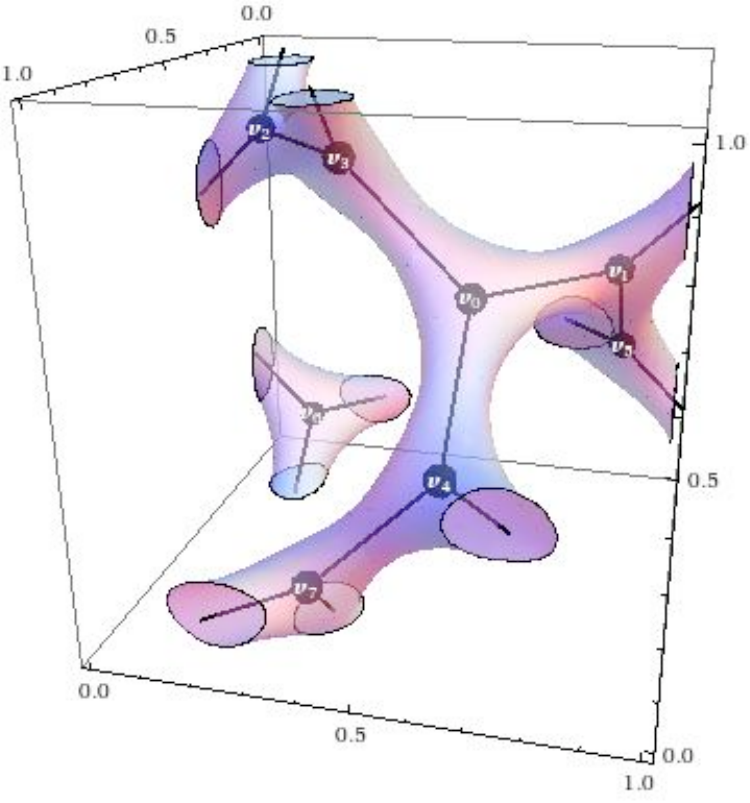}
\caption {\label{gyrlabelsfig} The graph $\Gamma_+$ embedded into its channel $C_+$}
\end{figure}

\begin{figure}
\includegraphics{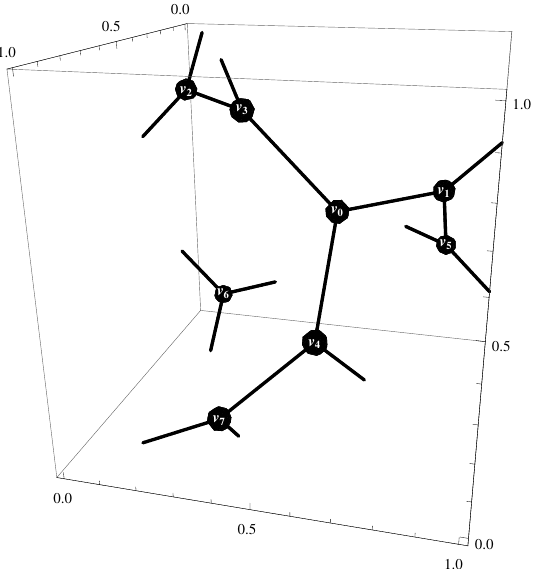}
\caption{\label{gammalabelsfig} The graph $\Gamma_+$}
\end{figure}

In constructing the surface, one can actually start
with one of these graphs and evolve the surface from it \cite{KGB}.
The symmetry of one skeletal graph is that of a single Gyroid $I4_132$.

To obtain the double Gyroid one can evolve from both the graphs.
The common symmetry group of both of these graphs is $Ia\bar3d$.
This symmetry group can actually be defined already on the nodes $V_{\pm}$ 
of the  two graphs $\G_{\pm}$. The subgroup $I4_132$ is 
then determined to be the  subgroup that fixes both sets of nodes (setwise).

In the case of the double Gyroid $S=S_1 \amalg S_2$, 
the complement $C=\R^3\setminus S$ of the 
disjoint union of the two Gyroid surfaces 
has three connected components. These are two channel systems $C_+$ and
$C_-$, each of which
can be retracted to its  skeletal graph $\Gamma_{\pm}$. 
In fact, the skeletal graph is a deformation retract \cite{KGB}.
There is a third connected
component $F$ and we will call $\bar F=F\cup S$.
This is a 3-manifold with two boundary components,
more precisely,  $\partial \bar F=S=S_1\amalg S_2$.
$F$ can be thought of as a ``thickened'' (fat) surface. The thickness
is fixed by the parameter $w$. Note that $F$ can
be retracted to each of the two boundary surfaces $S_i$.
In fact there is a deformation retract of $F$ onto the Gyroid.

This means that all homotopical information about the complements of the double
Gyroid is encoded in the Gyroid/level surface and the two skeletal graphs.

The translational symmetry group for both the Gyroid and the double
Gyroid is the Bravais lattice bcc.
Note that we will usually use ``lattice'' in the physical terminology, e.g.\
speak about the honeycomb lattice. The term ``Bravais lattice'' will be used
to denote a maximal rank mathematical lattice, i.e.\ a free rank $n$
Abelian subgroup of $\R^n$. In order to preempt any confusion, we give
precise definitions for our terminology in \S\ref{latdefsec}.

\subsection{The skeletal graph}

\subsubsection{The vertices and edges}
We will now describe the graph $\Gamma_+$ embedded into $\R^3$.

 Set
\begin{eqnarray}
v_0=(\frac{5}{8},\frac{5}{8},\frac{5}{8})&&v_4=(\frac{7}{8},\frac{5}{8},\frac{3}{8}) \nn\\
v_1=(\frac{3}{8},\frac{7}{8},\frac{5}{8})&&v_5=(\frac{1}{8},\frac{7}{8},\frac{3}{8})\nn\\
v_2=(\frac{3}{8},\frac{1}{8},\frac{7}{8})&&v_6=(\frac{1}{8},\frac{1}{8},\frac{1}{8}) \nn\\
v_3=(\frac{5}{8},\frac{3}{8},\frac{7}{8})&&v_7=(\frac{7}{8},\frac{3}{8},\frac{1}{8})
\end{eqnarray}
and let $\bar V_+=\{v_0,\dots, v_7\}$.

Furthermore $\Z^3$ acts on $\R^3$ by translations and we let
$V_{+}=\Z^3(\bar V_+)$ be the image of the set $\bar V_+$ under
this action. We will sometimes call this set of points the Gyroid lattice.

Given two points $v,w\in \R^3$ let $\overline{vw}=\{(1-t)v+tw|t\in
[0,1]\}$ be the line segment joining them. Also given a point $v$ we
let $T_x(v)=v+(1,0,0), T_y(v)=v+(0,1,0), T_z(v)=v+(0,0,1)$ be the
translated points.

 Consider the following
set of line segments

\begin{eqnarray}
\label{edgelist}  \bar E&=&\{\overline{v_0v_1}, \overline{v_0v_3},
\overline{v_0v_4},
\overline{v_2v_3},\overline{v_4v_7},\overline{v_1v_5},\nn\\
&&\overline{v_4\, T_x(v_5)},
\overline{v_7\, T_{x}(v_6)},\overline{v_1\, T_y(v_2)},\nn\\
&&\overline{v_5\, T_y(v_6)},\overline{v_2\,
T_z(v_6)},\overline{v_3T_z(v_7)}\}
\end{eqnarray}

Let $E_{+}=\Z^3(\bar E)$, where again $\Z^3\subset \R^3$ acts
as a subgroup of the  translation group.

\begin{df}
The  graph $\Gamma_+$ is the graph whose vertices are $V_{+}$,
whose edges are $E_{+}$ with the obvious incidence relations.
\end{df}

We recall:
\begin{prop}\cite{KGB}
$C_+$ can be deformation retracted onto the graph $\Gamma_+$ and
$\Gamma_+$ is the component of the critical graph contained in $C_+$.
\end{prop}

\begin{cor}
The homotopy type of the complement $T=\R^3\setminus G$ is the same
as that of two copies of $\Gamma_+$. In particular each channel has
the same homotopy type as $\Gamma_+$.
\end{cor}

This implies that all topological   invariants $C_{\pm}$ which are
homotopy invariant  are isomorphic to those of $\G_+$. In particular, this means
all homology, homotopy and K-groups of the topological space $C_+$ or $C_-$ are
determined on $\Gamma_+$.

\subsection{The quotient graphs}
Quotienting out by either the standard translation group or the bcc lattice,
we obtain the following two quotient graphs.

\subsubsection{Crystallographic Quotient Graph}
\label{cubepar}
 Let $\cube$ be the graph $\Gamma_+/\Z^3$ thought of as an abstract
 graph. This graph has a natural map embedding into the 3-torus
 $\R^3/\Z^3$.

\begin{prop}
$\cube$ is a cube embedded into the 3-torus. More precisely the
vertices of $\cube$ are $\bar V$ and the edges are the images
of $\bar E$ in $\R^3/\Z^3$.
\end{prop}

The abstract graph is given in Figure \ref{qubefig} which
also contains the  images of the vectors $e_i$ which name and orient all edges.

\begin{figure}
\includegraphics[width=0.6\textwidth]{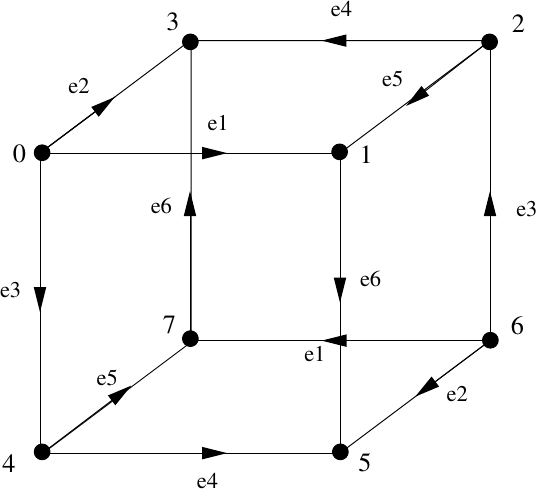}
\caption{\label{qubefig} The graph $\cube$ for the Gyroid}
\end{figure}

\begin{proof}
Since all the vertices in the fundamental domain do not lie on the
boundary, they give exactly the representatives of $V$ in
$\R^3/\Z^3$. The classes of the edges $E_{\Gamma_+}$ are given exactly
by the images of the set $\bar E$ in $\R^3/\Z^3$.
\end{proof}

\subsubsection{The (maximal) Quotient Graph}

By modding out by $\Z^3$, we have not yet used the full translational
symmetry of $\gp$ which is the bcc lattice. 
A set of generators for the bcc lattice  is 
\begin{equation}
\label{bccveceq}
f_1:=(1,0,0), \quad f_2=(0,1,0), \quad  f_3:=\frac{1}{2}(1,1,1)
\end{equation}
another set of generators which is more symmetric and we will use later on is:
\begin{equation}
\label{bccveceq2}
g_1=\frac{1}{2}(1,-1,1),\quad g_2=\frac{1}{2}(-1,1,1), \quad g_3=\frac{1}{2}(1,1,-1)\end{equation}

 We let $L=L(\gp)$ be the  free Abelian subgroup
of $\R^3$ generated by these vectors. 
We define   $\bar\Gamma_+$ to be the abstract
 quotient graph $\Gamma_+/L$.

\begin{prop}
$\bar\Gamma_+$ is the graph with 4 vertices and 6 edges, where 
all pairs of distinct vertices are connected by exactly one edge. 
\end{prop}
This graph is sometimes also called the complete square. Its incidence matrix
has entries one everywhere except on the diagonal, where the entries are zero.
This graph is shown in Figure \ref{square}.

\begin{figure}
\includegraphics[width=\textwidth]{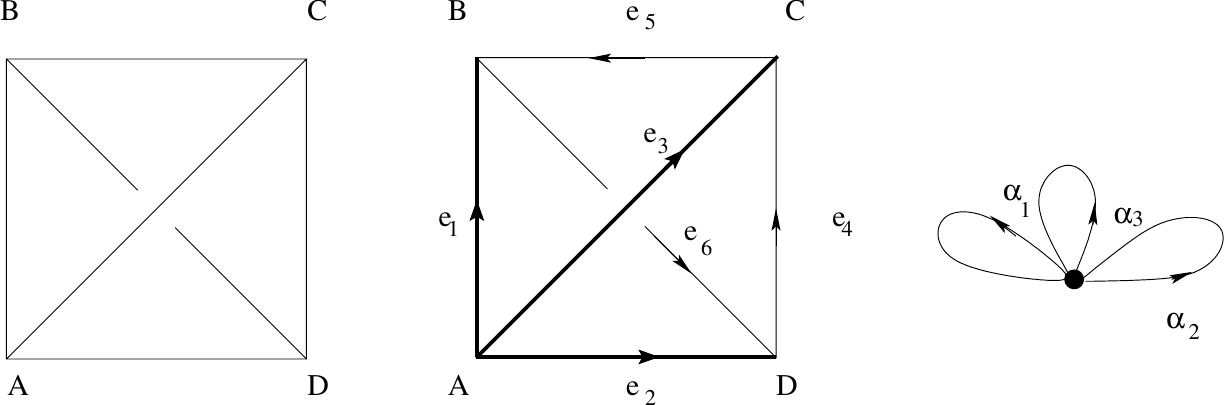}
\caption{\label{square} The complete square, the rooted spanning tree $\tau$ (root $A$ and edges $e_1,e_2,e_3$) and the vectors corresponding
to the oriented edges, the collapsed tree $\bar\Gamma_+/\tau$}
\end{figure}

\begin{proof}
We see that there is an embedding of $\Z^3\subset L$, 
where $L$ is the bcc lattice, so that we only have
to divide $\cube$ by the additional symmetry generated by
 $d:=\frac{1}{2}(1,1,1)$.  Now mod $\Z^3$, $T_d^2\cong id$ 
and $T_d$ (the translation by $d$) simply interchanges the vertices of $\cube$ as follows:
$v_0\leftrightarrow v_6, v_1\leftrightarrow v_7, v_2\leftrightarrow v_4$
and $v_3 \leftrightarrow v_5$. 
Hence we are left with 4 vertices, and we can choose the 
representatives $v_0,\dots, v_3$. Checking the list of edges 
(\ref{edgelist}) one sees 
that indeed the 12 edges form pairs and one can choose
the representatives $\overline {v_iv_j};i\neq j: i,j=0,\dots,3$. 
\end{proof}

Notice that this corresponds to a $\Z/2\Z$ symmetry of the graph
$\cube$. It is given by mapping each vertex to its diagonally
opposite vertex
and maps the
edges accordingly.

\subsection{The Underlying Group and Lattice}
There is another Bravais lattice hidden in the geometry of the Gyroid. This is the
fcc lattice. The nearest neighbor positions differ by vectors generating an fcc lattice.

This means in particular that  after shifting by $v_0$ the {\em positions} of the vertices 
of $\gp$ all lie on an fcc lattice.

In order to fix notation set:
\begin{eqnarray}
\label{gyvectorseq}
e_1=v_1-v_0 & e_2=v_3-v_0& e_3=v_4-v_0\\
e_4=v_3-v_2&e_5=v_7-v_4 &e_6=v_5-v_1
\end{eqnarray}

\begin{df}
Let $T(\gp)$ be the group of $\R^3$ that is generated by
\begin{eqnarray}
e_1=\frac{1}{4} \left(\begin{matrix}-1\\1\\0\end{matrix}\right),&
e_2=\frac{1}{4} \left(\begin{matrix}0\\-1\\1\end{matrix}\right),&
e_3=\frac{1}{4} \left(\begin{matrix}1\\0\\-1\end{matrix}\right)\\
e_4=\frac{1}{4} \left(\begin{matrix}1\\1\\0\end{matrix}\right),&
e_5=\frac{1}{4} \left(\begin{matrix}0\\-1\\-1\end{matrix}\right),&
e_6=\frac{1}{4} \left(\begin{matrix}-1\\0\\-1\end{matrix}\right)\\
\end{eqnarray}
\end{df}

\begin{prop}
The group $T(\gp)$ is isomorphic to $\Z^3$. The Bravais
lattice it generates is a face centered cubic (fcc).
\end{prop}

\begin{proof}
There are relations among the $e_i$ given as follows.
\begin{equation}
\label{relationeq} e_1=-e_5+e_6, \quad e_2=-e_4-e_6, \quad
e_3=e_4+e_5
\end{equation} so that we see
that $e_3,e_4,e_5$ generate. These vectors are linearly independent
over $\Z$ since they are independent over $\R$. Hence they provide a
free basis and an isomorphism to $\Z^3$. The vectors $e_4,-e_5,-e_6$
are the standard primitive vectors for the face centered cubic.
\end{proof}
Of course there are many other choices of basis here, e.g.\
$\{e_2,e_4,e_6\}$.

\begin{prop}
\label{fccprop}
The vertices of $\Gamma_+$ translated by $-v_6$ lie on the 3-dimensional face centered
cubic lattice generated by $e_4,e_5,e_6$.
\end{prop}
\begin{proof}
 $\gp$ is path connected
and  we take $v_6$ as the base point. Each line segments in $E$
corresponds to an edge. Choosing an orientation for this edge
defines a vector. Now, the statement follows from the fact that  the
vectors corresponding to the line segments in $E$ and hence those of
$E_{\gp}$ are exactly the vectors $\pm e_1,\dots,\pm e_6$.
\end{proof}

\subsection{Fundamental Groups, Loops and Effective Normal Vectors}
There is certain geometric information which can even already
be read off from the simple graph  $\bgp$. Such as its fundamental group or
the minimal loops starting at a given vertex.
 A loop on $\gp$      is
minimal
if it passes through a minimal number of edges. A more general treatment  will be given in \S\ref{latticesec}.

\begin{prop}
Let $\F_3$ be the free group in three variables.
The (realization) of the graphs $\gp$ and $\bgp$ have following fundamental groups:
\begin{enumerate} 
\item $\pi_1(\gp)=[\F_3,\F_3]$
\item $\pi_1(\bgp)=\F_3$
\end{enumerate}
in particular $\gp$ is the maximal Abelian cover of $\bgp$.

Since $\gp$ is homotopic to one channel these results hold for each channel of the Gyroid.
\end{prop}
\begin{proof}
We start with $\bgp$. This graph is homotopic to the wedge product of three $S^1$'s, whence
the second claim follows. 

In view of Lemma \ref{looplem} below (more generally by Proposition \ref{fundprop}),
 we see that  $\pi_1(\gp)$ is the subgroup  that 
is the kernel of the map $\F_3$ to its Abelianization. Indeed the sum of powers
of each of the generators in a word in the group has to be zero, which precisely means
that such a word is in the kernel of the Abelianization map or in other words, the commutator group.
\end{proof}
\begin{lem}
\label{looplem}
A loop on the graph $\bgp$  lifts to a loop on $\gp$ if and
only if each edge is traversed the same number of times in each
direction.
\end{lem}

\begin{proof}
 The ``if'' direction is clear since this means that the translations 
have to add up to zero. The equivalence
 follows from the general Proposition \ref{fundprop}, by noticing that indeed 
the lifts $l_1=e_2 e_6^{-1}e_1^{-1}, l_2=e_1e_5^{-1}e_3^{-1}, 
l_3=e_3e_4e_2^{-1}$ give rise to the vectors $\vec{l}_i=f_i$ of (\ref{bccveceq2}) 
and are  linearly independent.
\end{proof}

\begin{prop}
There are closed loops in the graph $\gp$.
 Each minimal loop goes through 10 sites
and at each point there are 30 oriented minimal loops or 15 such undirected loops. 
\end{prop}

\begin{proof}

A path in which one goes back and forth through an edge is
homotopic to the path in which this step is omitted. 
By direct calculation one can see that such a path is given by
traversing any five of the six edges. At the first step one has $3$
choices of edges, at the second step there are two and then again two possibilities.
Here one either returns to the original vertex or not. In the first
case there are 2 completions and in the second case 3 completions
to a minimal loop. Thus we have $2\cdot 3\cdot 5=30$ possibilities.
A detailed version of this enumeration is provided at the end of this section
in \S\ref{loopcalcsec}.
\end{proof}

\subsubsection{Explicit Calculation of the loops}
\label{loopcalcsec}

Here we give the details of the calculation of the loops. 
Using Proposition \ref{fundprop}, we have to look for paths that traverse each edge
the same number of times.
  We first look at the cases 
where each edge is traversed only once in each direction. 
We call such a path a good path.
We also have to keep the number of these edges minimal.
This already puts a simple constraint on the path. We may not go back and forth
 through the same edge. Starting at the given vertex we have to pass 
through 2 distinct edges. After this we have a choice, we can either go back
to the starting vertex (case I) or we can go to the only vertex which
we have not reached yet (case II). In case I the next oriented edge
is fixed, but then we have two choices Ia and Ib. After this we have
already used 5 edges so that the minimal possible 
 number of oriented edges and hence
the length of the loop is 10.  Indeed, we  can complete the edge 
path uniquely to a good path  without
traversing the 6th edge.

In case II the fifth oriented edge for a good path is fixed, going opposite
the first oriented edge.
There is a choice for the sixth oriented edge
in the path. Either returning to $v_o$ (IIb) 
or not (IIa). The case IIb has  a unique choice for a fifths oriented
edge for a good path and this has a unique 
completion to a good path involving 5 edges, again giving a length of 10 loops.
In case IIa, we again have two choices for the fifth oriented edge (IIa1) and
(IIa2). Both these choices have a unique completions to good paths again of length 10.

It remains to treat the cases 
where an edge is traversed more that once in each direction. 
Since we are not allowed to go back and forth on one edge
the choices for the first three oriented edges are as above. Now in Case I
we could go along the first oriented edge again, but this would lead 
us to traverse at least 6 edges and hence would not be minimal. At the next 
step of case I the choices are precisely Ia or Ib and traversing an edge twice in the same direction would not be minimal.
For case II, the first stage where one could use an oriented edge twice
for IIa is the fifth edge. But then one would need at 
least 6 edges counting multiplicities. After fixing the fifth oriented
edge one can only increase the number of traversed oriented edges by not choosing a good path. Finally in case IIb again the first edge with a choice to traverse an edge twice in one orientation is at the fifth oriented edge, but as before
this would lead to a path of length greater than 10.

So all in all we have $2\cdot 3=6$ choices for the first two edges and then
once these are fixed 5 choices for a good path. 
This means in all there are 30 such oriented paths.

\subsubsection{Explicit loops}
On the graph $\bgp$ there is a symmetry group of order 6 preserving
the base point and the spanning tree $\tau$  which permutes the vertices B,C and D.
This is precisely the group that gives us the 6 first choices.

In order to write down a shorter list, we will make the following observations.
Since the inverse of a minimal path is a minimal path, 
we can cut down the number to 15. Now since each of the paths traverses
five edges, it misses one. There are two cases: (1) the edge that is missed 
is incident to $v_0$, i.e.\ $e_1,e_2$ or $e_3$ and case (2) 
it is not, i.e.\ $e_3,e_4,e_6$. Case (1) corresponds to IIa1 and IIa2,
which case (2) corresponds to Ia, Ib and IIb.
In the case (1) the vertex $v_0$ is traversed 1 additional time except at
the start and end of the path and in case 
(2) it is traversed an additional 2 times.
This decomposes the loops into two pieces of length 5 in case (1)
 or three pieces of length 3,3,4; 3,4,3; 4,3,3 
in case (2) for Ia,Ib,IIb respectively. 
With each loop, its cyclic permutation of these components is also a loop.
There are 2 such loops in case (1) and 3 such loops in case (2). 
This also explains the 15 loops as $15=3\cdot 2 +3\cdot 3$ and permutes
the cases IIa1, IIa2; and Ia,IIb and Ib cyclically.

\begin{figure}
\includegraphics[width=.8\textwidth]{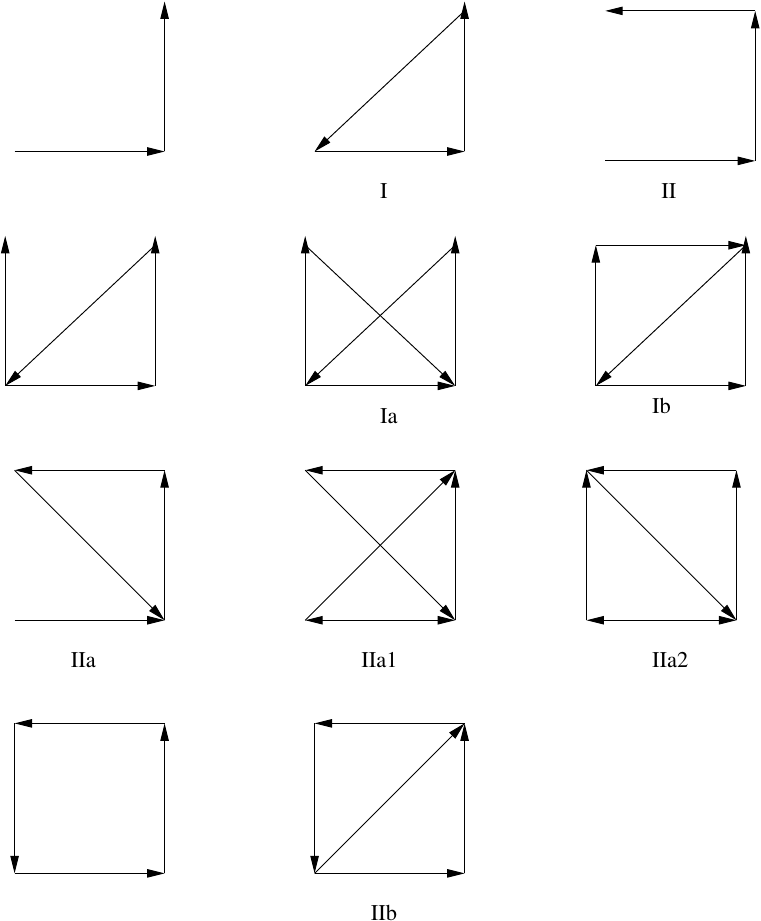}
\caption{\label{pathfig}
The combinatorial cases for the minimal loops of the Gyroid}
\end{figure}

We give here a list of six basic loops which may be useful for 
further discussion. To obtain all 30 one should cyclically permute the blocks
and take full inverses of the loops. We give the loops as an edge path as
well as their decomposition into the basis loops $l_1,l_2,l_3$.

\begin{table}
\begin{tabular}{llll}
loop as an edge path in $e_i${\it s} &in basis $l_i$& $N_{eff}$&edge missed\\
\hline
$e_2(-e_4)e_5e_6(-e_2)\;\; e_3e_4(-e_6)(-e_5)(-e_3)$&$(l_1l_2l_3)^{-1}(l_1l_2l_3)$&
$(1,-1,0)$&$e_1$ missing\\
$e_1(-e_5)e_4(-e_6)(-e_1)\;\; e_3e_5e_6(-e_4)(-e_3)$&$(l_2l_3l_1)(l_3l_1l_2)^{-1}$&
$(0,1,-1)$&$e_2$ missing\\
$e_1(-e_5)e_4(-e_6)(-e_1)\;\;e_2(-e_4)e_5e_6(-e_2)$&$(l_2l_3l_1)(l_1l_2l_3)^{-1}$&
$(1,0,-1)$&$e_3$ missing\\
$e_1(-e_5)(-e_3)\;\;e_2(-e_6)(-e_1)\;\;e_3e_5e_6(-e_2)$&$[l_2,l_1]$&
$(1,1,0)$&$e_4$ missing\\
$e_1e_6(-e_2)\;\;e_3e_4(-e_6)(-e_1)\;\;e_2(-e_4)(-e_3)$&$[l_1^{-1},l_3]$&
$(0,-1,-1)$&$e_5$ missing\\
$e_1(-e_5)e_4(-e_2)\;\;e_3e_5(-e_1)\;\;e_2(-e_4)(-e_3)$&$[l_2,l_3]$&
$(-1,0,-1)$&$e_6$ missing\\
\end{tabular}
\caption{\label{looptable}A generating set the shortest loops and their effective normal
vectors}
\end{table}

\begin{rmk}
It is interesting to note that the explicit isomorphism 
$\F_3 \to \pi_1(\bgp)$ given on the generators $\alpha_i$ of $\F_3$
by $\alpha_i\mapsto l_i$ induces a new length on $\F_3$ as the minimal
length in the letters $e_i$. It is this length that remembers the structure
of the Gyroid. The computation above is a good example of this. 
The three examples of commutators are minimal 
elements of $[\F_3,\F_3]$ which are
not the identity in the length in $\alpha_i$ or $l_i$. 
They are of length 4. 
In the $e_i$ their length is 10. The other examples are of 
length 6 in the $\alpha_i$, so not minimal in this metric, but they are 
minimal in the $e_i$ being of length 10.
\end{rmk}

\subsubsection{Effective normal vector}
One distinguishing feature of the different loops is their spatial orientation.
This for instance has an effect on the flux through a surface bounded by such a loop.

We assume a constant magnetic field. In this case if $S$ is a surface bounding 
a loop $L=e_{i_1},\dots, e_{i_{n}}$, let $f_j:= \sum_{k=1}^{j}e_{i_j}$. Then
by Stokes, we can just integrate over the surface given by the union of the triangles defined 
by $(f_j,f_{j+1}),j=1,\dots,n-1$. 
If $N_j=f_j\times f_{j+1}=f_j\times e_{j+1}$ and $N_{eff}:=\sum_{j=1}^{n-1}N_j$
then we have for the magnetic flux through $S$
\begin{equation}
\Phi =\iint_S B dS=\sum_j \frac{1}{2} B\cdot N_j=\frac{1}{2} B\cdot N_{eff}
\end{equation}

The values for $N_{eff}$ are listed for the basic loops. Notice that inversion of a loop inverts
the normal vector, while the cyclic permutation of the components leaves the effective
normal vector invariant.

\subsection{The quantum graph}
In order to promote the skeletal graph to a quantum graph, we will
fix a Hilbert space and a  Hamiltonian. Here we follow the terminology that a quantum graph
is graph with an associated Hilbert space and a Hamiltonian on it. In this paper we will use
the generalized Harper Hamiltonian \cite{belissard,Sunada}.

The original Harper Hamiltonian \cite{Harper} 
is obtained for a cubic lattice by using the 
tight-binding approximation and Peierls
 substitution for the quasi momentum \cite{PST}.

In the next section, we will give the general theoretical setup for this using graphs, groups and
representations.
For the Gyroid the notions such as graphs and groups 
have been introduced above, so that the reader may substitute
these in the general definitions below.

\section{Graph Representations and Matrix Harper
Operators}
\label{graphsec}

One idea in studying the non--commutative aspects of a given
system is to give a K--theoretic gap labeling for the Hamiltonians.
For this one considers an algebra $\mathscr B$ generated by the Hamiltonian
and the symmetries. If everything is commutative then this algebra
is basically  the $C^*$ algebra of functions on the Brillouin zone/torus.
We will make these ideas precise using physical and mathematical lattices as
defined below.

\subsection{Graph language}
For us a graph or an abstract graph $\Gamma$ will be a collection of vertices $V(\Gamma)$
and a collection of edges $E(\Gamma)$ which run between vertices
up to bijections preserving the incidences. Each edge can have
two orientations. An edge together with an orientation is called an
oriented edge. Each oriented edge $\vec{e}$ has a starting vertex $s(\vec{e})$
and a target vertex $t(\vec{e})$. A graph $\Gamma$
 is called finite if both $V(\Gamma)$ and  $E(\Gamma)$ are finite sets.

A graph naturally becomes a topological space if the edges are replaced by intervals. This space
is called the realization of the graph. In more technical terms the data above gives a one--dimensional CW complex and
we take the realization of this complex. When we talk about topological properties
of a graph, like its fundamental group, we always mean the topological properties of its realization.

A graph is connected if its realization is connected. This means that one can travel to all vertices from any given vertex along
the edges.
A tree is a connected graph whose realization is contractible (i.e.\ the graph has no loops). A choice of root of
a tree is simply a choice of a vertex and a rooted tree is a tree together with a choice of a root.
Given a graph $\G$ a subgraph $\tau$ is {\em called a spanning tree} if it is a tree and the vertices
of $\tau$ are all of the vertices of $\G$. To have a spanning tree $\Gamma$ needs to be connected.
In this case there are usually several choices of spanning trees. A rooted spanning tree is a spanning tree together with
the additional choice of a root.

\begin{prop}
\label{genloopprop}
Let $\bar \Gamma$ be a finite graph and $\tau$ be a rooted spanning tree. Let $v_0$ be the root then 
$\pi_1(\bar \Gamma):=\pi_1(\bar\Gamma,v_0)=\F_n$ where  $\F_n$ is the free group in $n$ variables and 
$n=|E(\bar \G)|-|E(\tau)|$.
\end{prop}

\begin{proof}
Consider the graph 
$\bar\G/\tau$ obtained by contracting all edges of the subgraph $\tau$. This 
is the graph which  only has one vertex $v_0$ and
all the edges are loops. The two graphs $\bar\G$ and $\bar\G/\tau$ are homotopy equivalent and hence have the same fundamental group.
 The graph $\bar\G/\tau$  embeds into the plane with $n$ punctures with each loop going around one puncture. 
To obtain a compatible embedding of $\G$ one blows up the 
only vertex of $\bar\G/\tau$ into the tree $\tau$ and both graphs are homotopy equivalent to the punctured plane.
$\bar \G$ is  thus homotopy equivalent to the wedge product of $n$ circles  $S^1$. 
It is well known that the first homotopy group of this space is the free group in $n$ generators.
\end{proof}

An embedded graph is a graph $\G$ together with an embedding of it realization into an $\R^n$.
Some properties we discuss depend on such an embedding or are derived from it.

\begin{ex}
We have so far considered  $\cube$ and $\bgp$ as abstract graphs and we have considered
$\gp$ both as an abstract as well as an embedded graph. The properties of having a certain number of  loops at a given point are properties of the abstract graph, while the proof made use of the embedding. 
The  effective normal vectors are properties of the embedding of the skeletal graph $\gp$ into $\R^3$.
\end{ex}

\subsection{Graph Harper Operator}
\label{harpsec}
\begin{df}
A $C^*$ representation $\grep$ of a graph $\Gamma$ is given by
the following data.
\begin{itemize}
\item A collection of separable Hilbert spaces $\H_{v}$ one
for each $v\in V(\Gamma)$.
\item A collection of isometries $U_{\vec{e}}:\H_{s(\vec{e})}\to
\H_{t(\vec{e})}$ for each oriented edge 
$\vec{e}$ such that $U_{\vec{e}}U_{\vec{e'}}=1$ whenever 
$\vec{e}$ and $\vec{e'}$ are the two orientations of the same edge.
\end{itemize}
\end{df}

\begin{rmk}
This construction can also be stated in more categorical terms. It is
a certain quiver representation. A finite graph $\Gamma$ also generates
a groupoid (that is a category in which all morphisms are invertible)
and in this setting a representation is a functor from the groupoid to the category of separable
Hilbert spaces.
\end{rmk}

\begin{df}
Given a representation $\grep$ of a 
finite graph $\bar\Gamma$ we set $\H_{\bar\Gamma}:=\bigoplus_{v\in V(\bar\Gamma)}
 \H_{v}$ and we define the graph Harper operator $H$ on $\H_{\bar\Gamma}$ as
\begin{equation}
H:=\sum_{\mbox{oriented edges } \vec{e} } U_{\vec{e}}  
\end{equation}
where we have (ab)used the notation $U_{\vec{e}}$ to denote the
partial isometry on $\H_{\bar\Gamma}$ induced by the operator of the same name.
\end{df}

\subsection{Matrix actions}
\subsubsection{General Setup}
\label{gensetsec}

Fix some  finite index set $I$, a fixed index $o\in I$ and and order on $I$, 
such that $o$ is the smallest element. 
Fix isomorphic Hilbert spaces $\H_i, i\in I$ and 
let  $\phi_i:\H_o\to \H_i$ be fixed choice of isomorphisms. 
We allow a choice of $\phi_o$. 
It might be that $\phi_o=id$ but this is not necessary. 

Set $\H=\bigoplus_{i\in I} \H_i$.

In this situation, there is an action of $M_{|I|}(End(\H_o))$ on $\H$
given as follows. Let $A\in M_{|I|}(End(\H_o))$ be an endomorphism
valued matrix. Then the action of $A$ on $\H$ is given by

\begin{equation}
\H=\bigoplus_{i\in I}\H_i\stackrel{\oplus_i \phi_i^{-1}}{\longrightarrow}
\H_o^{\oplus |I|}\stackrel{A}{\rightarrow} \H_o^{\oplus |I|}
\stackrel{\oplus_i \phi_I}{\longrightarrow}\H
\end{equation}

Vice--versa given an endomorphism $H\in End(\H)$ there is
a corresponding matrix $M(H)\in M_{|I|}(\H_o)$ given by
\begin{equation}
\H_o^{\oplus |I|}\stackrel{\oplus_i \phi_i}{\longrightarrow}\H
\stackrel{H}{\rightarrow} \H
\stackrel{\oplus_i \phi_i^{-1}}{\longrightarrow}\H_o^{\oplus |I|}
\end{equation}
That is if $\Phi=\bigoplus_i\phi_i$, $M(H)=\Phi^{-1}\circ H\circ \Phi$.

\subsubsection{Matrix action for the graph Harper operator}
In order to write the Harper operator of \S\ref{harpsec} as a matrix
according to \S\ref{gensetsec} we need to fix several choices. We will
assume that the graph $\Gamma$ is connected.
First there is a choice of base vertex $v_o$ and a choice of order on
all vertices in which the base vertex is the smallest.  Then for each vertex
$v$ we choose a fixed path $p_v$ ---that is a sequence of ordered edges
$\vec{e}_1,\dots,
\vec{e}_k$--- from $v_o$ to $v$. We then set 
$\phi_v:=\phi_{\vec{e}_k}\circ \cdots\circ \phi_{\vec{e}_1}:\H_{v_o}\to \H_v$.
If $v=v_o$ we also allow the choice $\phi_{v_o}:=id:\H_{v_o}\to \H_{v_o}$.
This corresponds to the empty path.

Given such a choice, we obtain the action 
$\rho_{p_1,\dots,p_n}$ as above. We call the resulting matrix {\em the
graph Harper operator in matrix form}.

 Given a rooted spanning tree $\tau$ of $\G$  choose $v_0$ to be
the root and then there is a unique shortest path on $\tau$ from $v_0$ to any vertex $v_i$ of $\tau$ which
we can view as a path on $\Gamma$. Thus we get the data needed. Again a convenient choice of paths is
given  by a spanning tree. We sometimes write $H_{\Gamma,\tau}$ for the corresponding
graph Harper operator in matrix form.

\begin{rmk}
Any two matrix Harper operators are conjugate. This follows simply by making a base change.
\end{rmk}

\section{$C^*$--geometry of Harper Hamiltonians on lattices }
\label{latticesec}
For lattices there are different definitions and connotations in the mathematics and
physics literature. We will use the adjectives ``mathematical'' and ``physical'' to distinguish
the two. As a reference in the physics literature for the definition of
lattices we use \cite{ashcroft}.

\subsection{Translation action}
The additive Abelian group
$\R^n$ acts on itself by translation. We denote $T_w(v)=v+w$. 
For a subset $S\in \R^n$ and a subgroup $L\subset \R^n$ we denote
by $L(S)$ the set of all translates of points in $S$ under
the action of all elements of $L$.

\subsection{Lattices: mathematical, physical and Bravais}
\label{maxsymsec}
\label{latdefsec}

\begin{df}
A {\em mathematical lattice with a basis} of rank $m$
in $\R^n$ is an injective group homomorphism  $\rho_L:\Z^m\hookrightarrow \R^n$. If $m=n$ then such a lattice is called
a {\em Bravais lattice with a basis}. 
\end{df}
If $\rho$ is a mathematical lattice with a basis,
we let $L:=\rho_L(\Z^m)$ be its image then
this is a discrete subgroup of $\R^n$ which is isomorphic to $\Z^n$.
 We will refer to $L$  as a mathematical/Bravais lattice. 

A Bravais subgroup or sublattice is a subgroup $L'\subset L$ which 
is a Bravais lattice.

Any Bravais lattice $L$  
acts by translations on $\R^n$ and a {\em primitive cell} 
 is a fundamental domain for this action. In general a
{\em cell} is a fundamental domain for a Bravais subgroup $L'$ of $L$. 

\begin{ex}
The
integer points $\Z^n\subset \R^n$ are a Bravais lattice which we call
the {\em crystallographic lattice}. A crystallographic cell is
a primitive cell for this lattice. The unit cube is a primitive cell. The cube
of length $2$ is a cell that is not primitive.
\end{ex}

As a definition for a lattice in the physical sense, we will take the
following convention:
\begin{df} A subset $\Lambda\subset \R^n$ 
 is a lattice if there is a finite set $V\subset\R^n$
and a Bravais lattice $L$  such that $\Lambda=L(V)$.
\end{df}
The tuple $(L,V)$ is called a crystal structure
or a lattice with a basis\footnote{Usually one would include
labels such as atom labels to $V$.}.
 We will not use the latter terminology as
it is mathematically confusing.

Notice that given $\Lambda$ neither $L$ nor $V$ are uniquely defined. 
We can and will mostly choose $L$ to be maximal and $V$ to be minimal 
or primitive. In this case we call $(L,V)$ a primitive crystal structure.
Sometimes $L$ is called the space group.

Still then there is of course a choice in fundamental domain and hence
a choice in $V$, but for any two such choices there is a unique
bijection. In fact if we let $\bar V$ be a set of orbits of
the $L$ action; then it is unique and any $V$ is just a
 choice of representatives. It is also easy to say when
 $L$ is minimal. $L$ is minimal if and only if the number
of orbits $|\bar V|$ is minimal.

\subsection{Graphs and Lattices}
To develop a general theory, we will deal with two cases. 
The usual case is that given a lattice we obtain a graph by
adding edges to the nearest neighbors. But is is convenient 
to also {\em allow} that we {\em already
have an embedded graph $\G_{\L}$ whose vertices are the set $\L$}.
In this case, we will  let $\G_{\L}$ stand for this chosen graph.

\subsubsection{Canonical Graph of a Lattice}
Given a lattice $\Lambda$  {\em without a pre--chosen graph} 
define the graph $\Gamma_{\Lambda}$  
of $\Lambda$ to be the graph whose vertices are the elements of
 $\Lambda$ and whose edges are the line segments between nearest
neighbors.
This graph is naturally embedded in $\R^n$ and we will sometimes
make use of this. 

\subsubsection{Quotient graphs}
Given a choice 
of symmetry group $L$  for $\Lambda$, 
we also define the quotient graph $\bar\Gamma_{\Lambda}(L)$ as 
$\Gamma_{\Lambda}/L\subset \R^n/L$. This is the
abstract graph whose vertices are given by $\bar V$ and whose
edges are given by the set of orbits of edges of $\Gamma_{\Lambda}$.

If $L$ is a maximal group we will just write $\bar\Gamma_{\Lambda}$.

If the unit cell is a cell for $L$ corresponding to a subgroup $L'$, 
we also define the {\em crystallographic quotient graph} of $\Lambda$
to be the graph $\bar\Gamma_{\Lambda}(L')$.

\subsection{Group of a lattice}
\label{grplatsec}
A lattice $\Lambda$ also generates a discrete subgroup of $\R^n$ as
follows: using the natural embedding $\Gamma_{\Lambda}\subset \R^n$
we can think of  any directed edge $\vec{e}$ of $\Gamma_{\Lambda}$
of $\Gamma_{\Lambda}$ simply as a vector in $\R^n$.
Then   there is an associated
translation operator $T_{\vec{e}}$. Notice that as vectors
the translates of all the oriented edges correspond to one another. 
This means that the set of {\em vectors} $\vec{e}$ is
given by one vector in each orbit of oriented edges of the embedded graph $\G_{\L}$. This set is
in bijection with the oriented edges of $\bar \Gamma_{\Lambda}$ and we will not distinguish between them in the notation. 
That is we write $\vec{e}$ for both the oriented edge of the abstract graph $\bgl$ and the vector in $\R^n$ it is in bijection with.
We will enumerate the set of vectors $\{\vec{e}\}$ as
$\vec{e_i}$. These
vectors then give a symmetric
group of generators of an Abelian subgroup of $\R^n$
which we call {\em the lattice  group $T(\Lambda)$}  of $\Lambda$.
It is the group generated by the $\vec{e}_i$
This group are all the elements of the form

\begin{equation}
T(\Lambda):=\{\sum_i a_i \vec{e}_i| a_i\in \Z\}
\end{equation}

\begin{prop}
$T(\Lambda)$ is a mathematical lattice.
\end{prop}
\begin{proof}
To generate the group we only need one choice of orientation per edge.
Let $\vec{e}_j, j \in J$ be such a choice. This choice defines a map
$\Z^J\to \R^n$. The image of this group is a torsion free Abelian
group and hence by the structure theorem for Abelian groups a free
Abelian group of rank $\leq |J|$.
\end{proof}

$T(\Lambda)$ also naturally acts by translations on $\R^n$.

\subsection{Fundamental group}
The graph $\G_{\L}$ is a covering space for $\bar\G_{\L}$. Since we know
the fundamental group of $\bar\G_{\L}$ by  Proposition (\ref{genloopprop})  this information
can be used to calculate the fundamental group of  $\G_{\L}$.
In order to do the computation, we will have to fix some notation. We fix a rooted spanning tree
$\tau$ of $\bar\G_{\L}$ and set $n:=|E(\bar \G_{\L})|-|E(\tau)|$. By contracting $\tau$ 
we obtain a surjection $\pi:\bgl \to \wedge_{i=1}^n S^1$ which by  Proposition (\ref{genloopprop}) induces
an isomorphism on fundamental groups.
 We fix generators of $\pi_1(\wedge_1^n S^1)$ 
and choose lifts $l_1,\dots, l_n$ of them.
By definition each loop $l_i$ is a sequence of directed edges $\vec{e}_{i_1},\dots,\vec{e}_{i_{k(i)}}$ for
some $k(i)$. We set $\vec{l_i}:=\vec{e}_{i_1}+\cdots+ \vec{e}_{i_{k(i)}}$.

\begin{prop}
\label{fundprop}
Let $\L$ be a lattice with graph $\G_{\L}$  and finite quotient graph $\bar\Gamma_{\L}$.
Let $\tau$ be a rooted spanning tree for $\bar\Gamma_{\L}$ and set $n:=|E(\bar \G_{\L})|-|E(\tau)|$.
 If the vectors $\vec{l}_i$ are linearly independent, 
then $\pi_1(\G_{\L})=[\F_n,\F_n]$.
Furthermore a loop on $\bgl$ lifts to a loop on $\gl$ if and only
if it traverses each edge the same number of times in each direction.
\end{prop}

\begin{proof}
We need to compute which loops on $\bgl$ lift to loops on $\gl$ and which do not. Given a loop $l$ on $\bgl$ it can be written
as a word in the basis loops $l_i$, say $w=\prod_j l_{i_j}^{\eps(j)}$ where $\eps(j)=\pm1$.
 Fixing a pre--image of $v_0$ given by the rooted spanning tree,
this lifts to an edge path on $\gl$ given by the sequences of vectors $\vec{e}_{i_1},\dots,\vec{e}_{i_{k(i)}}$.
Notice that $l_i^{-1}$ gives rise to the sequence  $-\vec{e}_{i_{k(i)}},\dots,-\vec{e}_{i_1}$.
 These vectors form a closed loop if and only if they all sum up to $\vec{0}$. Now each $l_i^{\pm1}$ 
contributes $\pm\vec{l_i}$ to the sum. This means that $w$ lifts to a loop
if and only if $\sum_j \eps(j)\vec{l_{i}}=\vec{0}$. Since
  by assumption the $\vec{l_i}$ are linearly
independent this happens if and only if for fixed $i$ the sum of the exponents of the occurring $l_i$ is zero,
and this happens precisely if the image of $w$ in the Abelianization $Ab(\F_n):=\F_n/[\F_n,\F_n]=\Z^n$ 
is $0$. This means that the covering group of the cover $\gl\to \bgl$ is $Ab(\F_n)$ which is
a normal subgroup of $\F_n$ and hence the covering group of $\gl$ is $[\F_n,\F_n]$.
The last statement follows immediately by noticing this is true
for the $\vec{l_i}$ and hence also for their summands $\vec{e_{i_j}}$.
\end{proof}

\subsection{Main examples}
\subsubsection{The honeycomb lattice} 
\label{honeydef}
The  honeycomb 
lattice (see Figure \ref{hondecfig}) is an example of a physical lattice.
To make things precise, we let  $\Lhex\subset R^2$ which is
the Bravais lattice generated by $-e_1:=(1,0)$ and 
$e_3:=\frac{1}{2}(1,-\sqrt{3})$. We set $e_2=-e_1-e_3=\frac{1}{2}(1,\sqrt{3})$.
A ``dual'' Bravais lattice $\Lhex^t$ generated by 
$f_2:=e_2-e_1=\frac{1}{2}(-3,\sqrt{3})$ and
$f_3:=e_3-e_1=\frac{1}{2}(3,\sqrt{3})$ acts via translation on $\Lhex$. 
There are precisely three orbits of this action. 
These are the $A$, $B$ and $C$ lattices,
where we fix the $A$ lattice to be the orbit of $(1,0)$, the $B$ lattice
to be the orbit of $(-1,0)$, and the $C$ lattice to be the orbit of $(0,0)$.
The lattice $\Lhc$ is then defined to be the union of
the $A$ and $B$ sublattices. In  $\Lhc$ 
there are three nearest neighbors as indicated in Figure \ref{hondecfig}.

The symmetry group is $\Z^2$ embedded as the ``dual'' Bravais lattice
above. The group $T(\Lhc)$ is
again isomorphic to $\Z^2$,
 but it is the original triangular lattice 
generated by $e_1$ and $e_3$. 
See Figure \ref{hondecfig}.
This defines the graph $\G_{\Lhc}$.

\begin{figure}
\includegraphics[width=.7\textwidth]{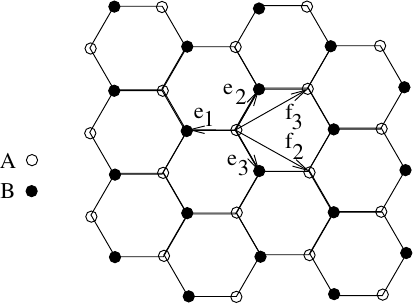}
\caption{\label{hondecfig} The honeycomb lattice $\Lhon$, the generators $f_i$
for $L(\Lhon)$ and the generators $e_i$ of $T(\Lhon)$}
\end{figure}

\subsubsection{The Gyroid lattice graph}
The Gyroid lattice is given by the set of vertices $V_+$ of $\gp$ of \S \ref{gclasssec}. It is also
a physical lattice. 
 Here the symmetry
group is the space group of $I4_132$ which is the body centered cubic (bcc) 
lattice.
The group $T(\L)$ is actually the face centered cubic (fcc) lattice
as shown in Proposition \ref{fccprop}.

\subsection{Hilbert space of a lattice}
\label{hilbpar}
Given a lattice or in general a countable set 
$\Lambda$, we define its Hilbert space
to be $\H_{\Lambda}:=l^2(\Lambda)$ that is the set of square
summable complex sequences indexed by $\Lambda$.
There is an alternative way to think about a sequence $(a_{\lambda}:\lambda\in
\Lambda)$ as a function $\psi$ on $\Lambda$. Given a sequence as above the
function
is given by $f(\lambda):=a_{\lambda}$. Vice--versa given $\psi$
the sequence is obtained by setting   $a_{\lambda}:=\psi(\lambda)$.

Given an action of a group $G$ on $\Lambda$ there is an induced action
of $G$ on $l^2(\Lambda)$ given by $g\psi(\lambda):=\psi(g^{-1}(\lambda))$
for $g\in G$ and $\psi\in l^2(\L)$.

A standard basis for
$l^2(\L)$ is given by the functions $v_l(l')=\delta_{l,l'}$ .

If $\L$ is a lattice then there is a unitary representation of $\L$ on $\H=l^2(\L)$,
given via the translation operators  $T_l(v_{l'})=v_{l+l'}$.

Given a lattice $\L$ and fixing a translation group  $L$ we can decompose the Hilbert space $\H_{\L}$ by
 breaking it down in terms of the orbits of $L$. More precisely, labeling each orbit by a vertex in $\bar\G$, we obtain
the direct sum decomposition
\begin{equation}
\label{dircompeq}
\H_{\Lambda}=\bigoplus_{\bar v\in \bar \Gamma} \H_{\bar v}
\end{equation}
where $\H_v=l^2(L(v))$ for any $v$ which represents the class of $\bar v$.

\subsection{Partial isometries and projections}
\label{partialisopar}
In general we get a representation of $L(\Lambda)$ on $\H_{\Lambda}$.
Notice that there is in general {\em no representation} of $T(\L)$ on $\H_{\L}$. We do have a partial action,
which is given by partial isometries. If $\vec{e}$ is an edge from $v$ to $w$ in $\bar \G$ then it induces an 
isomorphism $V_{\vec{e}}:\H_w\to \H_v$ simply by setting $V_{\vec{e}}\psi(l)=\psi(l-\vec{e})$. This induces
a partial isometry on all of $\H(\Lambda)$.

 Another way to describe this partial isometry
is as follows. Consider $\H_{\Lambda}$. Since $T(\Lambda)\supset \Lambda$ 
--after shifting $\Lambda$ so that $0\in \L$-- we have that $\H_{T(\L)}\supset \H_{\L}$ and moreover $T(\L)$ acts on $\H_{T(\L)}$ via translation
operators $T_{\vec{e}}$.
If $P$ is the orthogonal projection of $\H_{T(\L)}$ to $\H_{\L}$ then $V_{\vec{e}}=PT_{\vec{e}}P$. If
 furthermore $P_v$ is the projection onto $\H_v\subset H_{T(\\L)}$
then we can further decompose the action into components
$U_{\vec{e}}^{w,v}=P_vT_{\vec e}P_w$. Notice that in principle there could be a directed edge
$e'\neq e$ with $\vec{e'}=\vec{e}$ as vectors in $\R^n$. This happens for e.g.\ for the cube in \S \ref{cubepar}. In this case
there will be more components or super selection sectors to use physics terminology.

\subsection{Projective representations}

Up to now, we have insisted that the translation operators are a {\em bona fide} representation
of the translations groups.

 It turns out that to accommodate such physical data as a magnetic
field one should only expect a projective representation. This is also in accord with general
quantum theory, where representations are always only expected to be projective.

This is also one of the sources of non--commutativity. The second being that one does not expect
the symmetries of the system to commute with the Hamiltonian on the nose. One actually
has a choice either to preserve the commutativity of the symmetry operators or to preserve
that the symmetry operators commute with the Hamiltonian. We recall some standard facts from representation theory.

\subsubsection{Cocycles}
A $U(1)$  2--cocycle for a group $G$ written as $c\in
Z^2(G,U(1))$ is a map $c:G\times G\to U(1)$ such that 
\begin{equation}
c(u,v)c(uv,w)=c(u,vw)c(v,w)
\end{equation}
which is called a cocycle condition.

Notice that if $c$ is a $U(1)$ 2--cocycle for $G$ and $H\subset G$ is a subgroup then its restriction
$c|_{H}:H\times H \to U(1)$ is a $U(1)$ 2--cocycle for $H$.

\begin{df}
A morphism 
$\rho:G\to {\mathcal U}(\H)$, is called {\em a projective representation with cocycle $c$}
if $\rho(1_G)=id_{\H}$, that is the identity $1_G$ of $G$ maps to the identity operator
$id_{\H}$, and  $\rho(u)\rho(v)=c(u,v)\rho(u+v)$. 
\end{df}

\begin{ex}
A cocycle $c$ is called trivial if $c(u,v)=s(u)s(v)s^{-1}(uv)$ for some group morphism $s:G\to U(1)$.
Given a trivial group cocycle one can always perturb or scale an existing representation $\rho$ to a 
projective representation $\rho_c$ with cocycle $c$ by setting
$\rho_c(u):=s(u)\rho(u)$. With the same formula one can scale a projective representation with a cocycle $c'$ to one with the cocycle $cc'$.
\end{ex}

\begin{rmk}
We will be dealing with Abelian groups $G$, viz.\ $\Z^n$ for which we will use the usual additive 
notation of $0$ and $+$ for the identity and group operation.
\end{rmk}

In the case of a free Abelian group and its Hilbert space, given any cocycle,
 one can also twist the standard
representation to a projective representation with that cocycle.

\begin{lem}
Let $L\simeq \Z^n\subset \R^n$ be a lattice and  $\alpha \in Z^2(L,U(1))$. Then the operators
$U_{l}$ which operate on $\H_L$ via
\begin{equation}
U_l(v_{l'})=\alpha(l,l')v_{l+l'}
\end{equation} 
where $v_l$ is the standard basis satisfy

\begin{equation}
\label{comreleq}
U_lU_{l'}=\alpha(l,l')U_{l+l'},\quad
U_lU_{l'}=\eps(l,l')U_{l'}U_{l} \mbox{ with }
 \eps(l,l')=\frac{\alpha(l,l')}{\alpha(l',l)}
\end{equation}
\end{lem}

\begin{proof}
Straightforward calculation.
\end{proof}

\begin{lem}
Let $B$ be a bilinear form on $\R^n$. 
Then $\alpha_B(u,v):=e^{\frac{i}{2}B(u,v)}$ is a 2--cocycle for $\R^n$.
If $\L\subset \R^n$ is a lattice and $e_i$ are generators for this lattice
then the algebra of operators $U_l$ of the $\alpha$ twisted representation is generated
by the operators $U_i:=U_{e_i}$.
\end{lem}

\begin{proof}
First we calculate
\begin{eqnarray}
\alpha_B(u,v)\alpha_B(u+v,w)&=&\exp(\frac{i}{2}[B(u,v)+B(u+v,w)])\nn\\
&=&\exp(\frac{i}{2}[B(u,v)+B(u,w)+B(v,w)])\nn\\
\alpha_B(u,v+w)\alpha_B(v,w)&=&\exp(\frac{i}{2}[B(u,v+w)+B(v,w)])\nn\\&=&
\exp(\frac{i}{2}[B(u,v)+B(u,w)+B(v,w)])
\end{eqnarray}
secondly if $l=\sum a_i e_i$ then
$\prod U^{a_i}_i\propto U_l$ by the formulas (\ref{comreleq}).
\end{proof}

\subsubsection{Noncommutative Tori}
A standard example which will be important in the following is given by the
projective action of $\Z^n$ on $\H(\Z^n)$ with the cocycle $\alpha$ given by 
choosing an {\em anti--symmetric} bilinear form $B=2\pi \formTheta$ on $\R^n$ and then 
restricting the cocycle. This yields a cocycle which we call $\alpha_{\formTheta}:=\alpha_B(u,v)=e^{i\pi\formTheta(u,v)}$.

In this case the operators  $U_i:=U_{e_i}$ generate the algebra of operators $U_l$.
These generators satisfy

\begin{equation}
\label{ncteq}
U_iU_j=e^{2 \pi i \Theta_{ij}}T_jT_i, \quad  \Theta_{ij}=\formTheta(e_i,e_j)
\end{equation}
This follows from (\ref{comreleq}); $\eps(U_i,U_j)=e^{i\pi\Theta_{ij}}e^{-i\pi\Theta_{ji}}=e^{2\pi i \Theta_{ij}}$
by anti--symmetry of $\Theta$. In general:

\begin{df}
For a fixed $n\times n$ anti-symmetric matrix $\Theta$ 
the $C^*$--algebra generated by $n$ unitary operators $U_i$  on a separable Hilbert space
satisfying the commutation relations
(\ref{ncteq}) is called a non--commutative torus and denoted by $\T^n_{\Theta}$.
\end{df}

Note that $\T^n:=\T^n_0$ is the commutative
$C^*$--algebra corresponding to the torus $T^n=(S^1)^{\times n}$
under the Gelfand--Naimark theorem.

\begin{ex}
\label{twodex}
$n=2$: In this case the skew symmetric matrix can be written as
$\Theta= {\theta} \left(\begin{matrix}0&1\\-1&0\end{matrix}\right)$.
And
$\alpha(l,l')=  \exp[i\pi \theta (l\wedge l')]$, where for $l=(l_1,l_2)$ and
$l'=(l'_1,l'_2)$: $l\wedge l':=\det\left(\begin{matrix}l_1&l_1'\\l_2&l_2'
\end{matrix}\right)$ and accordingly  $\eps(l,l')=e^{2 \pi \theta l\wedge l'}$.
This case is written as $\T^2_{\theta}$.
\end{ex}
 
\begin{rmk}
As usual, once a basis $b_i$ for $\R^n$ is fixed there is a bijection between anti--symmetric
bilinear forms $\formTheta$ and skew--symmetric matrices $\Theta$ given
by $\Theta_{ij}=\formTheta(b_i,b_j)$. If thus choice of basis has been  made, we will write $\alpha_{\Theta}$. In our applications, the basis $b_i$ will be given to us by a choice of basis for the Bravais lattice $L$.
\end{rmk}

\subsubsection{Wannier or Magnetic translation operators}

\label{magneticfieldsec}

In case of magnetic field there is a standard cocycle and representation coming from the $B$--field.
 This was first used in \cite{Harper}. 
A magnetic field $B$  in mathematical terms is a 2--form
on $\R^n$.

If we assume that the $B$ field is constant 
then this is nothing but a skew--symmetric bilinear form on $\R^n$.
Thus giving rise to a cocycle $\alpha_B$ as above. This cocycle can now
be restricted to any lattice in $\R^n$.
Furthermore in this situation, we can choose a magnetic potential $A$. This is a 1--form
on $\R^n$ such that $dA=B$. This form exists, since $dB=0$ and as
$\R^n$ is contractible  $\bar H^*(\R^n)=0$, 
i.e.\ the reduced cohomology vanishes, so that every closed form is
exact.

In this case the cocyle is trivial and the twisted action  can be rewritten as

\begin{equation}
\label{neu}
U_{l'}\psi(l)=e^{-i \int_l^{(l-l')}A}\, \psi(l-l')
\end{equation}

\begin{rmk}

If $\sum_{i=1}^n m_i=0$ is a closed cycle of vectors bounding a
simply connected polygonal region $D$ with vertices $v_1,\dots ,v_n$, s.t.\ $v_{i+1}=v_i+m_i$ 
then $U_{m_1} \cdots U_{m_n}=e^{-i F}id$
where $F$ is the flux of $B$ through $D$. Furthermore if $B$ has a
non-singular vector potential $A$ such that $curl(A)=B$ then
$F=\int_D BdS=\sum_i \int_0^{1} m_i\cdot A(v_i+m_{i}t)  \, dt$

\end{rmk}

This means that given two elements $l,l'$ we have
that 

\begin{equation}
U_{l}U_{l'}=\eps(l,l')U_{l'}U_{l}, \quad \eps(l,l')=e^{i\int_R B dS}
\end{equation}
where $R$ is the rectangle spanned by 
the vectors $l$ and $l'$.

\begin{ex}
In the square lattice if 
we choose a constant magnetic field in $z$ direction $\vec {B}=\vec{k}$
or $B=2\pi \theta  dx\wedge dy$
then we can choose $A=\frac{1}{2}(\alpha_1 y\,dx + \alpha_2 x\, dy)$ with $2\pi \theta=\alpha_2-\alpha_1$.
Setting $U:=U_{e_1}$ and $V=U_{e_2}$ we obtain
$UV=e^{2\pi i\theta} VU$ and this is just $e^{i\phi}$ where $\phi=\int_D B dA$ is the flux
through the domain full square $A$ spanned by $e_1$ and $e_2$. The resulting algebra is $\T^2_{\theta}$.
\end{ex}

\subsection{Harper operator for a lattice $\Lambda$ and with graph $\Gamma_{\Lambda}$}
\label{harperlatsec}
In this section we construct a Harper operator for a given lattice $\Lambda$ and a graph for the lattice 
$\Gamma_{\Lambda}$. Associated to each such lattice there is a natural separable
Hilbert space $\H_{\Lambda}$. 
The Harper operator is an operator on $\H_{\Lambda}$ which
is obtained by giving a graph representation 
of $\bar \Gamma$. To do this we fix a maximal translational 
symmetry group $L$ for
$\Lambda$ (see \S \ref{maxsymsec}). Using the techniques 
of \S\ref{graphsec}
we obtain a Harper operator as we discuss in detail below. This operator together with the representation of the  translation group $L$ defines a subalgebra of the endomorphism algebra of
$\H_{\Lambda}$ which we call the Bellissard-Harper algebra and denote
by $\B(\G_{\L},L,\alpha)$ or $\B$ for short. Here we allow for
projective representations acting with the cocycle $\alpha$, e.g.\ via a magnetic field (see \S\ref{magneticfieldsec}). If $\G_L,L$ and $\alpha=\alpha_B$ with $B=2\pi i \formTheta$ is fixed,
we write $\BformTheta$ for the corresponding Bellissard--Harper algebra.

Furthermore again using  \S\ref{graphsec}, we can get a 
a matrix representation for the Harper operator and the action of $L$. 
These
matrices are elements of a matrix algebra with
coefficients in the algebra generated by the operators
corresponding to $L$.  Fixing $\G_\L,L$, a path basis,
 a basis for $L$ and a magnetic field $B=2\pi\formTheta$,
we arrive at an algebra of matrices, which we call $\BTheta$.
The details are given below.

\subsubsection{Graph representation of $\bar \Gamma$}

To put ourselves in the situation of \S \ref{gensetsec}, 
we will decompose $\H_{\L}$ as
\begin{equation}
\H_{\Lambda}=\bigoplus_{\bar v\in \bar \Gamma} \H_{\bar v}
\end{equation}
where $\H_v=l^2(L(v))$ for any $v$ which represents the class of $\bar v$.

We also fix a 2--cocycle $\alpha_B$ corresponding to a skew--symmetric bilinear form 
or equivalently a constant $B$--field. This gives us the Hilbert space part of the data of a graph
representation $\rho_{\bar \Gamma}$. We define the
isomorphisms as $U_{\vec{e}}$. For this we use the fact that
an oriented edge of $\bar \Gamma$ has a natural representation as
a vector $\vec{e}$ (see \ref{grplatsec}).

\begin{df}
The Harper operator $H_{\Lambda}$ is defined to be 
the graph Harper operator $H_{\bar \Gamma}$ on $\H_{\Lambda}$
corresponding
to the graph representation $\rho_{\Gamma}$.

The Bellissard--Harper (BH) algebra $\B(\G_\L,L,\alpha_B)$ 
is the $C^*$--algebra of operators on $\H_{\Lambda}$ generated
by the projective representation of $L$ and the graph Harper operator.
\end{df}

\begin{rmk}
The operator $\H_{\L}$ also generates random walks and is related to a discrete difference operator
as follows. Let $\gl$ be a $k$--regular graph, which means that each vertex has valence $k$.
Then $\Delta=k-H_{\L}$ acts as the difference operator. 
$\Delta(\Psi)(l)=(\sum_{\vec{e}:s(\vec{e})=l}\Psi(l)-\Psi(T_{\vec{e}}(l)))$ where the sum is over ``nearest neighbors'' as
defined by $\gl$.
\end{rmk}

\subsubsection{A Matrix representation of the Bellissard--Harper algebra}
\label{harpermatrixsec}
In order to get a matrix representation, we fix a vertex $v_o$ of $\bar
\Gamma_{\Lambda}$ and a choice of paths $p_v$ from $v_o$ to $v$. We
will call such a choice a choice of path basis. Again a convenient way to fix such
data is to specify a spanning tree.

We then get a matrix representation of the Harper operator and the
operators coming from the projective representation of $L$.

\begin{thm} Fixing a choice of path basis, and a basis for $L$ the corresponding faithful
matrix representation of $\B(\G_{\L},L,\alpha_{B})$ is a sub--$C^*$--algebra $\BTheta$ of the $C^*$
algebra $M_{|V(\bar\Gamma)|}(\T^{n}_{\Theta})$.
\end{thm}

\begin{proof}
Before passing to the matrix representation all the operators
involved are shifted translation operators. Those coming from $L$ and those
coming from $L(\Lambda)$. First we have to show that the operators from $L$ 
still act as operators from $L$ when restricted to $\H_{v_o}$, but this is clear
since these are diagonal in the direct sum decomposition
(\ref{dircompeq}). Thus the operators in question are conjugates,
$U_{p_v}U_lU_{p_v}^*\propto U_l$ for any $U_l\in \T^n_{\Theta}$.
Here $\Theta$ is the matrix obtained from $\formTheta=\frac{1}{2\pi}B$ by using the choice of basis of $L$.
Secondly, for $l'\in
T(\L)$,  $U_{p_v}U_{l'}U_{p_v'}^*\propto
U_{l''}$   acts as a translation operator which preserves the $v_o$ summand.
This means that the sum of vectors $l''=-p_v+l'-p_v'$ is actually in $L$. Hence the
assertion follows.
\end{proof}

Notice that different choices of path basis may lead to different
representations, but all these representations are isomorphic; moreover
they are conjugates of one another. The effect of  changing the basis of $L$
is to replace the matrix $\Theta$ with its basis transform $\Theta'$,
but as $C^*$--algebras $\T^n_{\Theta}=\T^n_{\Theta'}$ ---only the presentation has changed---
 with 
the base change acting as an endomorphism.

\begin{cor}
If $\Theta$ is rational then the spectrum of $H_{\L}$ has finitely many gaps.
Moreover the maximal number is determined by the entries of $\Theta$.
\end{cor}

\begin{proof}
Since there is an injection of $\BTheta$ into $M_{|V(\bar\Gamma)|}(\T^{n}_{\Theta})$ we can
restrict the tracial states to $\BTheta$. The image of the tracial states of $\T^{n}_{\Theta}$ is
known to be $S=\Z+\sum_{ij}\theta_{ij}\Z\subset \R$ \cite{E,EL}. We fix a faithful tracial state $\tau$ and
then we have that for any gap projection $P_{\lambda}$: $\tau(P_{\lambda})\in [0,|V(\bar\Gamma)|]\cap S$.
We thus see that there are only finitely many possible gaps if all the $\theta_{ij}$ are rational.
\end{proof}

\subsection{Geometry of  $\B$}
In general, we are given a lattice $\L$ and perhaps the graph $\G_{\L}$. We
can then obtain a family of BH--algebras by choosing different cocycles
$\alpha_{2\pi\formTheta}$. We will call an element of this family $\BformTheta$. Now we have
already shown that such an algebra has a faithful matrix representation $\BTheta\subset M_k(\T_{\Theta}^n)$ where $k$ depends on $\Gamma$.
It is interesting to note that this family of subalgebras has different geometries and K--theories
depending on the choice of $\Theta$. 
Generically one would expect
\begin{ep}
\label{expectation}
If $\Theta$ is generic (i.e.\ all entries are irrational) then $\BTheta=M_{|V(\bar\Gamma)|}(\T^{n}_{\Theta})$ which is Morita equivalent to $\T^{n}_{\Theta}$.
\end{ep}
If this expectation is met is of course dependent on the choices. 
It is true for all the cases we will study. The main motivation is that the non--commutative torus
at generic $\Theta$ is simple, i.e.\ there are no two-sided ideals. This usually
allows one to find that all the elementary matrices are in the algebra and hence
the algebra is the full matrix ring.  The details of our particular 
calculations given in \S 4 also illuminate this expectation.

An open question is what happens at non--generic values of $\Theta$, i.e.\ if one or more of the entries
of $\Theta$ are rational. This again heavily depends on the entries of $H_{\L}$. In the cases 
we study below either $\BTheta=M_{|V(\bar\Gamma|}(\T^{n}_{\Theta})$ again  or it is a genuine
subalgebra. This is for instance a good new source of such algebras and for families
in which the $K$--theory may jump.

The commutative case $\B_0$ is also very interesting. Here we can
characterize the $C^*$--algebra $\B_0$ by the space it represents via the Gelfand--Naimark theorem. 
For this we need some terminology.
For each character, or $C^*$--algebra morphism  $\chi:\T^n\to \C$  there is an induced
 $C^*$--algebra morphism $\bar \chi:M_k(\T^n)\to M_k(\C)$ for any $k$. We fix $k=|V(\bar\Gamma)|$. 
We say $H_{\L}$ is generic if it has $k$ distinct Eigenvalues in $\T^n$ or
equivalently if there is a character $\chi$ of $\T^n$ such that $\bar \chi(H)$
has $k$ distinct Eigenvalues.

We call a point $\chi$ of  $\T^n$ degenerate if
$\bar \chi(H)$ has Eigenvalues with higher multiplicities.
The action of $L$ gives an inclusion: $i:\T^n\to \B_0$.
Given a character of $\B_0$,  we also get
a character of $\T^n$ by pull--back along $i$. We call a point $\chi$ 
of $\B_0$ degenerate
if $\overline{i^*(\chi)}(H)$ has Eigenvalues with higher multiplicity.

\begin{thm}
\label{commthm}
If $H_{\L}$ is generic then $\B_0=C^*(X)$ where $X$ is a generically 
$k$--fold cover of the torus $T^n$. This cover is ramified over the locus of 
degenerate points and 
is moreover a quotient of the trivial $k$--fold cover. 
Here the identifications are along degenerate points of $X$.
\end{thm}

\begin{proof}
First notice that the trivial $k$--fold cover of $T^n$ has the  $C^*$--algebra $\T^n[e_1,\dots, e_k]/R$
where the $e_i$ are self--adjoint and generate a semi--simple algebra. 
This means that the
relations $R$ are equivalent to the equations  $e_ie_j=\delta_{ij}e_i$ 
and $\sum e_i=1$. Here the $e_i$ can be understood as the projectors
to each of the copies.

Secondly we characterize $\B_0$. It is certainly a quotient of the $C^*$ algebra $\T[H]$ where $H$ is a new
self--adjoint generator which commutes with the previous generators. 
The way to understand its quotient is as follows.
By the theorem of Caley--Hamilton we know that the characteristic 
polynomial $p$ of $H$ annihilates $H$: $p(H)=0$.
We claim this is the only relation. Indeed if there were any other relation $r$, then we could write
$r=r'+r''$ with $r'\in (p(H))$ and $r''$ a polynomial in $H$ of degree less than $k$.
This relation would hold after applying any character $\chi$, i.e.\ $\bar\chi(r)=0$. 
Since $\bar\chi(r')=0$ we also get that $\bar\chi(r'')=0$.
 But generically  there are $k$ distinct
Eigenvalues, so that for generic choices of $\chi$ the minimal polynomial of $\bar\chi(H)$ 
is the characteristic polynomial. Thus as a polynomial in $H$ the degree of $r''$ must be bigger
or equal to $k$ which means that $r''=0$.
Hence $\B_0=\T^n[H]/p[H]$.

We can now give the $C^*$ morphisms corresponding to the geometric continuous maps
\begin{equation}
\xymatrix
{
\bigsqcup_1^k T^n \ar[rr]^{\phi}\ar[dr]_{\pi_1}& &X\ar[dl]^{\pi_2}\\
 &T^n&
}
\end{equation}

Let $\lambda_i\in \T^n$ be the Eigenvalues of $H$ as a matrix with coefficients in the
commutative ring $\T^n$.
Then the map $\phi$ is given by $H\mapsto \sum \lambda_ie_i$.
The maps $\pi_1$ and $\pi_2$ are just the inclusion maps. There are sections $s_k$ of $\pi_1$ given
by sending $e_j\mapsto \delta_{j,k}1$ which give rise to sections $\tilde s_k=\phi \circ s_k$.
The corresponding $C^*$ map is given by $H\mapsto \lambda_k$. From this the claims follow readily.
\end{proof}

\section{Results for the Bravais, Honeycomb and Gyroid Lattices}

\subsection{The Bravais lattice cases}
\subsubsection{The $\Z^n$ case}
In case $\Lambda=\Z^n$, we see that $L=T(\L)=\Z^n$ and $\bgl$ is the
graph
with one vertex and $n$ loops. 

From the graph $\bgl$ we can already  read off that $\pi_1(\gl)=[\F_n,\F_n]$
according to Proposition \ref{fundprop}, 
since the condition is obviously satisfied.
Minimal loops are of length $4$ and there are $n \choose 2$ unoriented loops.

Fixing a  cocycle $\alpha$ by fixing
an anti--symmetric matrix $\Theta$  (recall that this
is equivalent to fixing a constant $B$--field),
the corresponding Harper operator is
just
\begin{equation}
H_{\Z^n}=\sum_{i=1}^n U_{e_i}+U^{-1}_{e_i}=\sum_{i=1}^n U_i+ U_i^*
\end{equation}
where $e_i$ are the standard unit basis vectors of $\R^n$ and 
$U_i$ are the generators of the non--commutative $n$--torus
$\T^n_{\Theta}$.

The algebra generated by the representation of $L=\Z^n$ is just
the non--commutative torus $\T^n_{\Theta}$ and since
 $H\in \T^n_{\Theta}$ the algebra $\BTheta$ is also the non--commutative
torus $\BTheta=\T^n_{\Theta}$.

\subsubsection{Other Bravais Lattices}
If $\Lambda$ is the set of points
of a Bravais lattice, then again $L=T(\L)=\L$. For the
graph $\bar \Gamma_{\Lambda}$ we need the information, which of the
distances between vertices of $\Lambda$ are minimal or we need the additional
data of a graph $\gl$. This information is also crucial in determining $\pi_1$.

Let $e_j:j\in J$ be
the collection of these minimal vectors and fix an orientation
$\vec{e}_j$ for each of them. As $\L=L$ is a subgroup $0\in \L$ and
the minimal vectors are given by the $\lambda\in \Lambda$ with minimal length.
Again fix  a cocycle $\alpha$ by fixing
the  anti--symmetric matrix $\Theta$.
Then 
\begin{equation}
H_{\Lambda}=\sum_{j\in J} U_{\vec{e}_j}+U^*_{\vec{e}_j}
\end{equation}
 
The algebra generated by $L$ is always
$\T^n_{\Theta}$ and since $H_{\L}\in\T^n_{\Theta}$ we again obtain $\BTheta=\T^n_{\Theta}$.

\begin{ex}
The triangular lattice. This is the lattice spanned
by the vectors $e_1=(1,0)$ 
and $e_2=(\frac{1}{2},\frac{\sqrt{3}}{2})$ in $\R^2$.
 In this case the graph $\bgl$ has one vertex and three loops
 with
the six oriented edges corresponding to $\pm e_1, \pm e_2, \pm (e_1-e_2)$.
Hence the condition of Proposition \ref{fundprop} is not met. One can compute the fundamental
group by elementary methods.

The choice of $\Theta$ is given by $\Theta= \theta \left(\begin{matrix}0&1\\-1&0\end{matrix}\right)$.
The Harper operator is given by
\begin{eqnarray}
H_{\Lambda_{\tiny{\mbox{tri}}}}&=&U_{e_1}+U^*_{e_1}+U_{e_2}+U_{e_2}^*+U_{e_1-e_2}+U^*_{e_1-e_2}\nn\\
&=&U_1+U_1^*+U_2+U_2^*+U_3+U_3^*
\end{eqnarray}
where $U_1U_2U_3=e^{i \pi \theta}id$.
Since $H\in \T^2_{\theta}$
the BH algebra is $\BTheta=\T^2_{\theta}$ in the notation of Example \ref{twodex}.
\end{ex}

\subsection{The Honeycomb Lattice}

\subsubsection{Classical geometry}
The honeycomb lattice $\L=\Lhon$ is the two--dimensional lattice specified in \S\ref{honeydef}.
Its quotient  graph $\bgl$ is the graph with two vertices 
and three edges depicted in Figure \ref{thetafig}.

\begin{figure}
\includegraphics[width=.8\textwidth]{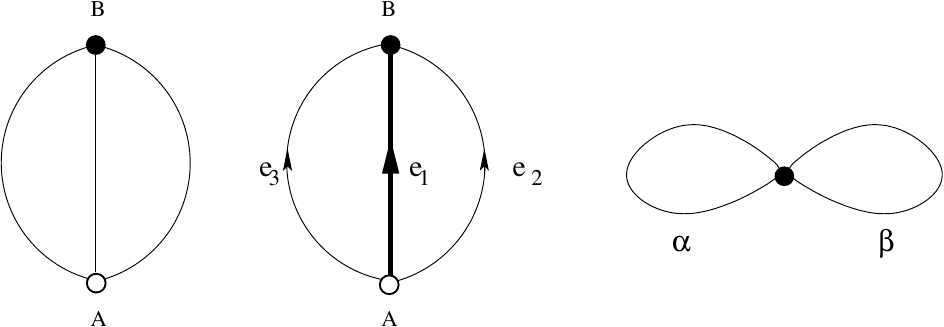}
\caption{\label{thetafig} The graph $\bgl$ for the honeycomb lattice, the spanning tree
 $\tau$ (root $A$, sole edge $e_1$) together with a set of oriented edges, the quotient graph $\bgl/\tau$}
\end{figure}
We set $e_1=(-1,0)$, $e_2=\frac{1}{2}(1,\sqrt{3})$ and $e_3=\frac{1}{2}(1,-\sqrt{3})$.
We choose the rooted spanning tree $\tau$ (root $A$, sole edge $e_1$ as indicated 
in Figure \ref{thetafig}). 
The loops $\alpha$ and $\beta$ lift
as $l_2= e_2(-e_1)$ and $l_3=e_3(-e_1)$, we see that $\vec{l_2}=\frac{1}{2}(-3,\sqrt{3})$ and $\vec{l_3}=\frac{1}{2}(3,\sqrt{3})$. Thus the condition of Proposition \ref{fundprop} is met
and we obtain
\begin{prop}
\begin{equation}
\pi_1(\gl)=[\F_2,\F_2]
\end{equation}
At each point of the lattice there are six directed or three undirected minimal loops of length six.
\end{prop}

\begin{proof}
It is clear by inspection that there are no loops of length two and four 
and the loops 
through an oriented edge twice will give length bigger than six.
To get a minimal loop, we need to pass through each 
edge exactly once in each direction, if possible. 
This can be done by fixing the first edge, and then
 choosing among the two left over edges  the second
returning edge. Now everything is fixed. 
One has to leave the vertex by the last edge and one concludes
a cycle of 6 edges by choosing the only possible non--traversed edge oriented 
without using an oriented edge twice. This of course gives the known three
unoriented loops at each vertex.
\end{proof}

\begin{rmk}
In the generators $\alpha,\beta$, these loops are given by the following
elements: $[l_2,l_3]^{\pm1}, [l_2^{-1},l_3]^{\pm1}$ 
and $[l_2,l_3^{-1}]^{\pm 1}$.

It is interesting to note 
that the map $\bgl\to \bgl/\tau$ induces a new 
length function on the free group $\F_2$ that is not symmetric in 
symmetric generators. E.g.\ the commutator $[l^{-1}_2,l_{3}^{-1}]$ has
length $8$ in this metric whereas all other commutators listed above
have length $6$ in the generators $e_i$.
\end{rmk}

\begin{rmk}
There is another observation. Given a loop on $\gl$ it decomposes into blocks
of loops on $\bgl$. Here we could take the 
loop $l_1=e_1(-e_2)\; e_3(-e_1)\; e_2(-e_3)$ 
which decomposes into 3 blocks. Now
cyclically permuting these blocks, we also get a loop on $\gl$ and of course
the inverse of any loop is a loop, so that
 the loop $l_1$ generates all six loops. 
\end{rmk}

\subsubsection{Quantum geometry}
Fixing a constant magnetic field $B=2\pi \formTheta$ amounts
to choosing an anti--symmetric bilinear form $\formTheta$ on $\R^2$.
We will use the induced cocycle both on $L$ and on $T(\L)$.
The matrix $\Theta$ is the matrix of $\formTheta$ with respect to the basis $(f_2,f_3)$ of $L$.
It is completely determined by its value $\formTheta(f_2,f_3)$. The cocycle induced on 
$T(\L)$ will also play a role. It is fixed by the value $\formTheta(e_1,e_2)$.
We fix the notation that 
\begin{equation}
\label{notationeq}
\theta:=\formTheta(f_2,f_3), \quad q:=e^{2\pi i \theta} \quad\text{\it and} \quad
\phi=\formTheta(-e_1,e_2), \quad \chi:=e^{i\pi \phi},
\end{equation}
The values of $\chi$ and $\theta$ are not independent:
\begin{equation}
\theta=-3\phi\quad  \text{\it and} \quad q=\bar\chi^6
\end{equation}
This follows from the elementary calculation 
$\theta=\formTheta(f_2,f_3)=\formTheta(e_2-e_1,e_3-e_1)=
\formTheta(e_2,e_3)+\formTheta(e_2,-e_1)+
\formTheta(-e_1,e_3)+
\formTheta(-e_1,-e_1)=\formTheta(e_2,-e_1)-\phi+\formTheta(-e_1,-e_2)=-3\phi$.

The Hilbert space $\H=l^2(\Lhon)$ splits 

\begin{equation}
\label{hondecomp}
\H_{hon}=\H_A\oplus \H_B
\end{equation}

where $A$ and $B$ are the vertices as indicated in Figure \ref{thetafig}.
The six oriented edges  $\pm e_i,i=1,\dots 3$ of $\bgl$ give rise to three
partial isometry operators $U_{\pm i}$.

$$U_i:=
 \left(\begin{matrix}
0&0\\
U_{e_i}&0\\      \end{matrix} \right), \quad 
U_{-i}:=
\left(\begin{matrix}
0&U_{-e_i}\\    
 0&0\\ \end{matrix} \right)
$$
where $U_{e_i}$ and $U_{-e_i}=U_{e_i}^{-1}=U_{e_i}^{*}$ are the isomorphisms as in \S\ref{partialisopar}.

The Harper Hamiltonian now reads:

\begin{equation}
H_{\L}=\sum_{i=1}^3 U_i+ U_i^{-1}= \left(\begin{matrix}
0&U_{e_1}^{*}+U_{e_2}^{*}+U_{e_3}^{*}\\
U_{e_1}+U_{e_2}+U_{e_3}&0\\      \end{matrix} \right),
\end{equation}

In order to put it into a matrix form with coefficients in $T^2_{\theta}$,
we again choose the spanning tree $\tau$  as indicated in Figure \ref{thetafig}.

The Hamiltonian now reads

\begin{equation}
H_{\L,\tau}=\left(\begin{matrix}
0&1+U^*+V^*\\
1+U+V&0
\end{matrix}
\right)
\in M_2(\T^2_{\theta})
\end{equation}
where we have used the operators $U:=\chi U_{f_2}$ and $V=\bar\chi U_{f_3}$. We have
that 

\begin{equation}
\label{uvcomeq}
UV=q VU \mbox { or  }UVU^*=qV
\end{equation}

The symmetry algebra generated by the translations $U_{f_i}$ is isomorphic to
$\T^2_{\theta}$ on $\H_A\oplus\H_A$. It acts via the representation defined by
$\rho(U_{f_2})=diag(U_{f_2},\chi^2U_{f_2})=diag(\bar\chi U,\chi U)$ and 
$ \rho(U_{f_3})=diag(U_{f_3},\bar\chi^2U_{f_3})=diag(\chi V,\bar \chi V) $.

\begin{prop}
\label{hpropnc}
If $q \neq \pm 1$ then $\BTheta=M_2(\T^2_{\theta})$ and
hence
is Morita equivalent to $\T^2_{\theta}$.
\end{prop}

\begin{proof}
The strategy of proof is to show that the elementary matrix $E_{12}\in \BTheta$.
In case this happens, we get that all elementary
matrices are in $\BTheta$, since $E_{21}=E_{12}^*, E_{11}=E_{12}E_{21}$ and 
$E_{22}=E_{21}E_{12}$ and then $\BTheta= M_2(\T^2_{\theta})$.

The method to obtain $E_{12}$ is by direct calculation using the commutation relations
(\ref{uvcomeq}).

The first step is to set $X=\rho(\bar\chi^2U)\rho(\bar\chi^2V^*)H\rho(V)\rho(U^*)$
and then set $X_1=\frac{1}{1-\bar q}[H-\bar q X]$. Note we are using the assumption
$q\neq 1$.
$$
X_1=\left(\begin{matrix}0&\frac{1}{1-\bar q}[(1-\bar\chi^2)+(1-\bar\chi^4)U^*+(1-\chi^4)V^*\\1&0\end{matrix}\right). 
$$
In step two, we set $X_2=H-X_1\rho(1+U+V)=\left(\begin{matrix}0&**\\0&0\end{matrix}\right) $,where explicitly
$$
**=BU^*+CV^*+DU+EV+FU^*V+GV^*U
$$
with 
$$
B=C=\frac{\chi^4(1-\chi^2)}{1-\chi^6},
D=\frac{1-\chi^2}{1-\chi^6},
E=\frac{\bar\chi^2(\bar\chi^2-1)}{1-\chi^6}, 
F=\frac{\chi^2(1-\bar\chi^4)}{1-\chi^6 },
G=\frac{\chi^2(\chi^4-1)}{1-\chi^6}
$$
notice that $\chi^6=\bar q \neq 1$ by assumption.
And the coefficients do not vanish if $\chi^4\neq 1$ and $\chi^2\neq1$.
But $\chi^2=1$ implies $\chi^6=1$ and $\chi^4=1$ implies $\chi^{12}=1$ and 
so $\chi^6=q=\pm 1$.
So if $q\neq \pm 1$ then all the coefficients are non-zero.
 We can obtain $E_{12}$ in several steps
by setting $X_3=X_2-\rho(U)X_2\rho(U^*), X_4=X_3-\rho(V)X_3\rho(V^*)$,
 $X_5=\rho(U)X_4\rho(U^*)-qX_4$ 
and finally obtain $E_{12}=X_5-\rho(\frac{1}{(\bar q-q)(1-\bar q^2)}U^*V)X_5$,
where for the last step we need the assumption $q\neq -1$.
\end{proof}

The situation for $q=-1$ is more complicated. Notice
that in this case $\T^2_{\frac{1}{2}}$ is not simple. 
For instance there is a $*$ homomorphism of $\phi:\T^2_{\frac{1}{2}}\to
\Cliff(-Id_2) \otimes\C$
where $\Cliff(-Id_2)\otimes \C$ is the Clifford algebra over $\C$ of the
standard quadratic form given by the negative of the $2\times 2$ identity matrix $Id_2$. If $i$ and
$j$ are the usual basis vectors then $\phi(U)=i,\phi(V)=j$.

We let ${\mathscr J}:=ker\phi=<1+U^2,1+V^2>$.

There is also an algebra involution $\wedge$ on $\T^2_{\frac{1}{2}}$ given by $\hat U=-V,\hat V=-V$.

\begin{prop} 
If $q=-1$ and $\chi^4\neq 1$  then  $\BTheta=M_2(\T^2_{\frac{1}{2}})$.
If $q=-1$ and $\chi^4=1$ then $\BTheta$ is the subalgebra of $M_2(\T^2_{\frac{1}{2}})$
given by matrices of the form
\begin{equation}
\label{algeq}
\left(
\begin{matrix}
a&b\\
\hat b&\hat a
\end{matrix}\right)
+J \quad a,b\in \T^2_{\frac{1}{2}}, J\in M_2({\mathscr J})
\end{equation}
\end{prop}

\begin{proof}
In the case $q=-1$ we can at first proceed as in the proof of 
Proposition \ref{hpropnc}.
As $X_1$ we obtain 
$$X_1=\left(
\begin{matrix}
0&\frac{1}{2}[(1- \bar\chi^2)+(1-\chi^4)(U^*+V^*)]\\
1&0
\end{matrix}\right)
$$
Here the two cases split dependent on whether $\chi^4=1$ or $\chi^4\neq 1$.

We will deal with the case $\chi^4\neq 1$ first.
We set $\tilde X_2=\tilde X_1-\rho(U)\tilde X_1\rho(\chi^2 U^*)$, 
$\tilde X_3=\tilde X_2+\rho(U)\tilde X_2\rho(\chi^2 U^*)$. 
Then $\frac{1}{(1-\chi^4)(1+\chi^4)}\tilde X_3=E_{21}$
and hence we get $E_{12}=E_{12}^*\in \BTheta$ and hence $\BTheta$ is the full matrix
algebra.

In case $q=-1$ and $\chi^4=1$ then since $q=\bar\chi^6=-1$, we get 
$\bar\chi^2=-1$ and  $X_1=I=E_{12}+E_{21}\in \BTheta$.
Furthermore, we get that 
$$
X_2=
 \left(
\begin{matrix}
0&1+U^*+V^*-U-V]\\
1&0
\end{matrix}\right),
$$
Setting 
$Y_{3}:= \frac{1}{2}\rho(U)[X_2+\rho(U)X_2\rho(-U^*)]$ and 
$\tilde Y_3= \frac{1}{2}\rho(V)[X_2-\rho(U)X_2\rho(-U^*)$
we get that
$$Y_3= \left(
\begin{matrix}
0&1+U^2\\
0&0
\end{matrix}\right)\in \BTheta,
\quad
\tilde Y_3= \left(
\begin{matrix}
0&1+V^2\\
0&0
\end{matrix}\right)\in \BTheta
$$

Let ${\mathscr B}'$ be the algebra above given in  (\ref{algeq}). 
It is easy to check that
$\B'$ is indeed a $C^*$ algebra since $M_2({\mathscr J})$ is a two
sided  ideal
of $M_2(T^2_{\frac{1}{2}})$ and $\wedge$ commutes with $*$.
Now $U^*+U$ and $V^*+V$ are both in ${\mathscr J}$
so that 
$$
H=
\left(\begin{matrix}
0&1+U^*+V^*\\
1-U^*-V^*&0\end{matrix}\right)
+\left(\begin{matrix}
0&0\\
U+U^*+V+V^*&0\end{matrix} \right)\in \B' 
$$
Furthermore certainly the operators of $L$ are in
$\B'$ so that the inclusion $\BTheta\subset \B'$ holds.

On the other hand, since $L\subset \BTheta$ all the matrices  $diag(a,\hat a)$  are in $\BTheta$
and by the above $I\in \BTheta$  so that  all matrices of the form $\left(
\begin{matrix}
a&b\\
\hat b&\hat a
\end{matrix}\right)$ are in $\BTheta$.
Furthermore since  the matrices $(1+U^2)E_{21}$ and analogously
$(1+V^2)E_{21}$ are
in $\BTheta$, taking products with $I=E_{12}+E_{21}$, we
obtain that all the matrices of  $ M_2({\mathscr J})\subset \BTheta$. Hence
that $\B'\subset \BTheta$ as claimed.
\end{proof}

We now analyze the case of $q=1$. Let $(x)$ be the principal (algebra) ideal generated by
$x$ in $\T^2_0$. Set $J_1=(1+U+V)$ and $J_2=(1+U^*+V^*)$ and $J_{12}=J_1J_2$.

\begin{prop}
\label{hpropcomm}
If $q=1$ and $\chi=\pm 1$ then 
$\BTheta=C^*(X)$ where $X$ is the double cover of the
torus $S^1\times S^1$
ramified at the points $(e^{2\pi i \frac{1}{3}},e^{2\pi i
  \frac{2}{3}})$
and $(e^{2\pi i \frac{1}{3}},e^{2\pi i
  \frac{2}{3}})$.

If $q=1$ and $\chi\neq \pm 1$ 
then $\BTheta$ is equal to the matrix $C^*$--subalgebra of $M_2(\T^2_0)$
given by matrices in the $C^*$--subalgebra

\begin{equation}
\label{almostcomeq}
\rho(\T^2)+
\left(\begin{matrix}J_{12}&J_2\\J_1&J_{12} \end{matrix}\right)
\end{equation}
\end{prop}

\begin{proof}
The first statement follows from Theorem \ref{commthm} for $\chi=1$. 
The Eigenvalues of $H$ in
 $\T^2_0$ are $\pm\sqrt{(1+U+V)(1+U^*+V^*)}$, which exist by continuous 
operator calculus as the operator
under the square root is self--adjoint operator of the form $AA^*$ 
and hence has non--negative real spectrum.
The points of the space $X$ equivalent to $\B_0$ where these two Eigenvalues coincide
are given by precisely the points above.

For $\chi=-1$, we remark that $\BTheta=\B_0$, by using the involution $U\mapsto -U, V\mapsto -V$.

In case that $q=1$, but $\chi\neq \pm1$, we actually know 
that $\chi$ is a sixth root of unity and  $\chi^2=\zeta_3=e^{2\pi i \frac{1}{3}}$
or $\chi^2=\zeta_3^2$ and in these cases $\chi^4=\chi\neq 1$.
Let $\B'$ be the algebra of (\ref{almostcomeq}). Since $J_1^*=J_2$ and both
of them are algebra ideals it is clear that $\B'$ is a $C^*$--subalgebra. 
Then it is clear that $\BTheta\subset \B'$.
In order to prove the reverse inclusion, we notice that 
$\frac{1}{1-\chi}[H-\rho(U)H\chi^2\rho(U^*)]=(1+U+V)E_{21}\in\BTheta$. Hence also 
$(1+U^*+V^*)E_{21}\in\BTheta$ as well as $(1+U+V)(1+U^*+V^*)E_{11}\in \BTheta$ and 
$(1+U+V)(1+U^*+V^*)E_{22}\in \BTheta$. This together with the action of $L$ shows, $\B'\subset \BTheta$. 
\end{proof}

In the notation of equation (\ref{notationeq}):
\begin{thm}
If $q\neq\pm 1$ 
the algebra $\BTheta$ has the $K_0$ group
$\Z\oplus \Z$. If $\chi=\pm 1$ then $K_0(\BTheta)=\Z^3$.
\end{thm}

\begin{proof}
For $\alpha\neq 1 $ this directly follows from Proposition
\ref{hpropnc} below. For the commutative case $q=1,\chi=\pm 1$ 
this follows from
the fact that the double cover of the torus obtained by identifying
two pairs of points of two tori
has the corresponding $K_0$ group.
\end{proof}

\subsection{The Gyroid}

\subsubsection{The lattice and sublattice Hilbert spaces}
The lattice $\gp$ has the Hilbert space $\H_{\gp}=l^2(\gp)$. 

The subspace $\H_{\gp}$ decomposes naturally into  subspaces $\H_i$
where $\H_i=\H_{\Z^3(v_i)}$ and $\Z^3(v_i)$ denotes the set of all
translates under $\Z^3$ of $v_i$. 

Thus
$\H_{\Gamma}=\bigoplus_{i=0}^7 \H_i$.
This is the composition corresponding to $\cube$.

In order to write  down the Hamiltonian effectively, we will put
together the summands in pairs.
\begin{equation}
\label{decompeq}
 \H_{\Gamma}=\H_A\oplus \H_B\oplus \H_C\oplus \H_D
\end{equation}

\begin{equation}
\H_A=\H_0\oplus \H_6, \quad \H_B=\H_1\oplus \H_7, \quad
\H_C=\H_3\oplus \H_5, \quad \H_D= \H_2\oplus \H_4
\end{equation}

This corresponds to passing from $\cube$ to $\bgp$ and
to our general setup of graph Hamiltonians.

\subsubsection{Cocycles for $\R^3$}
As discussed, choosing a magnetic field corresponds to a projective representation
of $\gp$ by unitary operators on $\H_{\gp}$ with a cocycle $\alpha$. 
We recall this here in a more familiar
three--dimensional setting. 

A skew symmetric bilinear form $B$ translates to a more well known expression as follows.
Let $i,j,k$ be the standard unit vectors on $\R^3$ and set $B_x=B(j,k)$, $B_y:=B(k,i)$
and $B_z=B(i,j)$ then if $\vec{B}=(B_x,B_y,B_z)$
\begin{equation}
B(m,m')=\vec{B}\cdot(m\times m'),\quad
\alpha_B(m,m')=\exp(i\frac{1}{2}\vec{B}\cdot
(m\times m'))\end{equation}

Indeed $\alpha_B$ is a  cocycle since
\begin{multline*}
B\cdot(m\times m')+B\cdot ((m+m')\times m'')
= B\cdot(m\times
m'+m\times m''+m'\times m'')\\
=B\cdot (m\times (m'+m''))+B\cdot (m'\times m'')
\end{multline*}

The physics interpretation of this is that $B$ is a fixed
constant magnetic field and then $\alpha_B(m,m')$ is the magnetic flux
though the {\em triangle} spanned by $m$ an $m'$.

\subsubsection{Cocycles for the Gyroid}
The maximal translation group $L$ for $\gp$ is the bcc lattice spanned by the vectors $f_i$ given in (\ref{bccveceq}) 
or $g_i$ (\ref{bccveceq2}). The lattice group
$T(\gp)$ is the fcc lattice spanned by the vectors $e_4,e_5,e_6$.
By restricting a cocycle given by a vector $\vec{B}$ on $\R^3$, we obtain
a cocycle for each of these lattices.
Since both the fcc and the bcc lattices are Bravais lattices, 
these are precisely all cocycles coming from cocycles based on anti--symmetric forms $\formTheta=\frac{1}{2\pi}B$ stemming from a magnetic field $B$.

We will use the basis $(g_i)$ to fix the matrix representation $\Theta$ of $\formTheta$.

\subsubsection{Graph representation and Graph Harper Operator on $\H_{\gp}$}
 We fix a vector $B$. This fixes the corresponding
 cocycle $\alpha(v,w):=e^{i\frac{1}{2}B\cdot(v\times w)}$.
Now consider the graph representation as defined in \S\ref{harperlatsec},
where the order of the vertices is $A,B,C,D$. 

Using the fixed cocycle $\alpha$ above,
we obtain the partial isometries corresponding
to elements of $T(\gp)$. These are $U_i:=U_{e_i}$ according to the list  
(\ref{gyvectorseq}) as 
discussed in \S\ref{partialisopar}.

In  Hilbert space decomposition (\ref{decompeq}) the Graph Harper Operator
$H_{\bar\Gamma}$ becomes the $4\times4$ matrix.

\begin{equation}
H_{\bgp}=\left(
\begin{matrix}
0&U_1^*&U_2^*&U_3^*\\
U_1&0&U_6^*&U_5\\
U_2&U_6&0&U_4\\
U_3&U_5^*&U_4^*&0
\end{matrix}
\right)
\end{equation}
recall that $U_i=\rho(e_i)$ and $U_i^*=U_i^{-1}=U_{-e_i}$.

\subsubsection{Matrix Harper operator}
We choose the rooted spanning tree $\tau$ (root $A$, edges $e_1,e_2,e_3$) as indicated in Figure \ref{square}.
Using this we obtain the following matrix Harper operator according to \S\ref{harpermatrixsec}

\begin{equation}
H_{\bgp,\tau}=\left(
\begin{matrix}
0&1&1&1\\
1&0&U_1^*U_6^*U_2&U_1^*U_5U_3\\
1&U_2^*U_6U_1&0&U_2^*U_4U_3\\
1&U_3^*U_5^*U_1&U_3^*U_4^*U_2&0
\end{matrix}
\right)
=:\left(
\begin{matrix}
0&1&1&1\\
1&0&A&B^*\\
1&A^*&0&C\\
1&B&C^*&0
\end{matrix}
\right)
\end{equation}
The coefficients can be expressed in terms of the operators of the magnetic translation operators
of the bcc lattice. We fix $U:=U_{f_1}, V:=U_{f_2}$ and $W:=U_{f_3}$ for the $f_i$ listed in (\ref{bccveceq}).
Then
\begin{equation}
A=aV^*W,\quad B=bWU^*,\quad C=cW^*UV
\end{equation}
where
$a=\frac{\alpha(e_2,-e_6)\alpha(e_2-e_6,-e_1)}{\alpha(-f_2,f_3)}$, 
$b=\frac{\alpha(-e_3,-e_5)\alpha(-e_3-e_5,e_1)}{\alpha(f_3,-f_1)}$ 
and $c=\frac{\alpha(-e_2,e_4)\alpha(-e_2+e_4,e_3)}{\alpha(-f_3,f_1)\alpha(-f_3+f_1,f_2)}$.

\subsubsection{Choices and Notation}

In order to proceed we fix some convenient basis and notation.
First we base change from the basis $f_i$ to the basis 
\begin{equation}
g_1=-f_2+f_3,\quad g_2=-f_1+f_3,\quad g_3=f_1+f_2-f_3
\end{equation}
 for the bcc lattice $B$. We see that $A=a'U_{g_1},B=b'U_{g_2},C=c'U_{g_3}$ 
again for fixed constants $a'=\alpha(-e_1,-e_6)\alpha(-e_1-e_6,e_2),b'=\alpha(-e_3,-e_5)
\alpha(-e_3-e_5,e_2)$ and $c'=-\alpha(-e_2,e_4)\alpha(-e_2+e_4,e_3)$. The operators
$A,B,C$ again generate the $C^*$ algebra of 
the  non--commutative torus $\T^3_{\Theta}$ where $\Theta$
is the matrix of the bilinear form in the basis $g_i$.
Explicitly:
\begin{equation}
 \theta_{12}=\frac{1}{2\pi}B\cdot (g_1\times g_2), \quad \theta_{13}=\frac{1}{2\pi}
 B\cdot (g_1\times g_3), \quad 
\theta_{23}=\frac{1}{2\pi}B\cdot (g_2\times g_3)
\end{equation}
We will fix the notation
\begin{equation}
\alpha_1:=e^{2\pi i \theta_{12}}=\alpha^2(g_1,g_2),
\bar \alpha_2:= e^{2\pi i \theta_{13}}=\alpha^2(g_1,g_3),
\alpha_3:=e^{2\pi i \theta_{23}}=\alpha^2(g_2,g_3)
\end{equation}
This means for the commutators
\begin{equation}
AB=\alpha_1 BA, \quad AC=\bar\alpha_2CA, \quad BC=\alpha_3CB
\end{equation}

\subsubsection{Matrix action of $T(\gp)$}

Given the choice of the spanning tree $\tau$ also determines the action of $T(\gp)$.

It will be convenient to introduce forth roots of the $\alpha_i$:

\begin{equation}
\phi_1=e^{\frac{\pi}{2} i \theta_{12}}, \quad \phi_2= e^{\frac{\pi}{2} i \theta_{31}},\quad
\phi_3= e^{\frac{\pi}{2} i \theta_{23}}, \quad \Phi=\phi_1\phi_2\phi_3
\end{equation}
and the matrices
\begin{equation}
\Lambda_1=diag(1,\phi_1,\phi_2,\overline{\phi_1\phi_2}), \quad
\Lambda_2=diag(1,\phi_1,\overline{\phi_1\phi_3},\phi_3),\quad
\Lambda_3=diag(1,\overline{\phi_2\phi_3},\phi_2,\phi_3)
\end{equation}
By definition $\alpha_1\alpha_2\alpha_3=\Phi^4$.

Then the action $\rho$ is given by the following matrices
\begin{equation}
\rho(A)=\Lambda_1diag(A,A,A,A), \rho(B)=\Lambda_2diag(B,B,B,B),\rho(C)=\Lambda_3diag(C,C,C,C)
\end{equation}

We make the following observation that makes the calculations quite a bit simpler.
\begin{equation}
\Lambda_1\Lambda_2^*=diag(1,1, \bar\Phi,\Phi),\quad \Lambda_2\Lambda_3^*=diag(1,\Phi,\bar\Phi,1),\quad \Lambda_1\Lambda_3^*=diag(1,\Phi,1,\bar\Phi)
\end{equation}

\subsubsection{Calculation of $\BTheta$ and $K_*(\BTheta)$}.
\begin{prop}
If $\Phi \neq 1$ then the BH--algebra $\BTheta$  is the full matrix
algebra $\BTheta=M_4(\T^3_{\Theta})$. 
\end{prop}

\begin{proof}
The strategy is to again show that enough elementary matrices are in $\BTheta$.
This will be done in a case by case study. We will present the first case in detail.

Assume that $\alpha_1\neq \Phi^2, \alpha_2\neq \Phi$ and $\alpha_3\neq \Phi$. Then we
have to do 6 steps to obtain $E_{34}$. These are 
\begin{eqnarray*}
X_1&=&H-\rho(AB*)H\rho(BA^*)\\
X_2&=&X_1-\rho(BC^*)X_1\rho(CB^*)\\ 
X_3&=&X_2-\rho(AC^*)X_2\rho(CA^*)\\ 
X_4&=&\Phi\bar\alpha_1 X_3-\rho(AB^*)X_3\rho(BA^*)\\
X_5&=& \bar\Phi\alpha_3X_4-\rho(BC^*)X_4\rho(CB^*)\\
X_6&=&\Phi\bar\alpha_2X_5-\rho(AC^*)X_5\rho(CA^*)
\end{eqnarray*}

The resulting matrix is 
$$(1-\bar\Phi^2\alpha_2\alpha_3)(1-\Phi\bar\a_3)(1-\bar\Phi\a_2)(\Phi\bar\a_1-\bar\Phi^2\a_2\a_3)
(\bar\Phi\a_3-\Phi\bar\a_3)(\Phi\bar\a_2-\bar\Phi\a_2)E_{43}$$
and the factor is invertible by assumption.
This provides $E_{43},E_{34}=E_{43}^*,E_{33}=E_{34}E_{43}$ and $E_{44}=E_{43}E_{34}$.
Now to get the other elementary matrices, we first assume $\phi_1\neq 1$ then set
$Y_1=HE_{44}-\bar\phi_3\rho(C)E_{34}$ and 
$Y_2=\phi_3Y_1-\rho(B^*)Y_1\rho(B)=\phi_3(\bar\phi_1-1)E_{24}$, 
and $E_{14}=Y_1-\bar\phi_{3}\rho(B) E_{24}$ which guarantees that all the $E_{ij}\in\BTheta$. 
If $\phi_1=1$ then we use 
$Z'_2=\bar\phi_3Y_1-\rho(C)H\rho(C^*)=\bar\phi_3(1-\alpha_3\bar\Phi)E_{24}$, 
so since $\alpha_3\neq \Phi$ 
we obtain $E_{24}$ and $E_{14}$ (as above) and hence again the whole matrix algebra.
The arguments for the cases $\alpha_1\neq \Phi, \alpha_2\neq \Phi^2,\alpha_3\neq \Phi$
and $\alpha_1\neq \Phi, \alpha_2\neq \Phi,\alpha_3\neq \Phi^2$ are analogous.

Now by pure logic it follows that $\BTheta$ is the full matrix in all cases except the three cases
$\alpha_i=\Phi^2,\alpha_j=\alpha_k=\Phi$ for $\{i,j,k\}=\{1,2,3\}$.
We treat only one of these, as the other two follow symmetrically. So assume
$\alpha_1=\Phi^2,\alpha_2=\alpha_3=\Phi$. In this case the story is similar, but faster.

We keep $X_1,X_2$ as above.  Then if $\Phi\neq \bar\Phi$ we
set $X'_3=\Phi-\rho(BC^*)X_2\rho(CB^*)=(1-\Phi)(1-\bar\Phi(\Phi-\bar\Phi)E_{13}$ and again
obtain the full matrix algebra using $HE_{11}$ to obtain the needed elementary matrices.
If $\Phi=\bar \Phi$ then since $\Phi\neq 1$ we have $\Phi=-1$ and hence $\a_2=\a_3=-1$
and so $\phi_2\neq \bar\phi^2$ since else $\phi_2=\pm1$ and $\a_2=\phi_2^4=1$.
Thus we get that $X_3''=\bar\phi_2X_2-\rho(A)X_2\rho(A^*)=4(\bar \phi_2-\phi_2)E_{13}$ with the
factor being invertible. And again, we are done.  
\end{proof}

\begin{prop}
If $\Phi = 1$  and not all $\alpha_i=\bar\alpha_i=\pm1$ then 
the BH--algebra $\BTheta$  is equal to the full matrix
algebra $\BTheta=M_4(\T^3_{\Theta})$. 
\end{prop}
\begin{proof}
As before this follows from an explicit representation of the generators $E_{ij}$ as expressions in $\BTheta$.
Since not all $\alpha_i=1$ there must be some $\alpha_i\neq 1$.
We will assume that $\alpha_3\neq 1$. The other calculations are symmetric.
Again one has two cases. Either $\alpha_2=\bar\alpha_2$ or not.
In the second case, we set $Y_1=H-\rho(BC^*)H\rho(CB^*)$ and 
$Y_2=
\alpha_2Y_1=\rho(AC^*)Y_1\rho(AB^*)$.We set 
$Y_3=\alpha_3\bar\phi_1\bar\phi_3^2Y_2-\rho(C)H\rho(C^*)$, and 
$Y_4=\bar\phi_1\phi_2Y_3-\rho(A)H\rho(A^*)=(
1-\bar\alpha_3)(\alpha_3-\bar\alpha_3)
(\alpha_3\bar\phi_2\bar\phi_3^2-\bar\phi_2^2\phi_3)
(\bar\phi_1\phi_2-\a_2-\bar\phi_2\bar\phi_2^2)E_{43}$.  
Now as $\alpha_3\neq 1$:  $Y_4\neq0$ if $\phi_2\neq 1$. In case $\phi_2=1$,
we must have $\phi_1\neq 1$ because otherwise $\phi_3$ would also be $1$.
In this case, we proceed analogously to obtain $E_{24}$. 
In both subcases we obtain the full matrix algebra following the 
strategy used in the previous proof.

Finally, we deal with the case that $\alpha_2=\bar \alpha_2$: 

If $\alpha_2=1$
then $\alpha_1=\bar\alpha_3$ then it follows that $\alpha_1\neq \bar \alpha_1$
and we are done by an analogous calculation. Indeed if $\alpha_1=\bar\alpha_1$
then also $\alpha_3=\bar \alpha_3$ and all $\alpha_i$ are real which we excluded.

If $\alpha_2=-1$ then $\alpha_1=-\bar\alpha_3$ and it follows that $\alpha_1\neq \bar \alpha_1$ and we are done by an analogous calculation. Indeed if $\alpha_1=-\bar\alpha_3$ then $\alpha_1=\bar\alpha_1$ again means that all three $\alpha_i$
are real which was excluded.
\end{proof}

\begin{prop}\label{phitwoprop}
If $\Phi = 1$ and all of the $\alpha_i=1$  then 
the BH--algebra $\BTheta$  is a proper subalgebra of the full matrix
algebra $\BTheta\subsetneq M_4(\T^3)$. 
\end{prop}

\begin{proof}
We will show this using the  character $\chi:\T^3\to \C$, defined by 
$\chi(A)=\chi(B)=\chi(C)=1$. This character induces an algebra morphism
$\bar\chi:M_4(\T^3)\to M_4(\C)$ 
and we will show that 
$\bar\chi(\BTheta)\subsetneq M_4(\C)$ which implies the result.
We note that 
$$
\bar\chi(H)=
\left(\begin{matrix}
  0&1&1&1\\
1&0&1&1\\
1&1&0&1\\
1&1&1&0
\end{matrix}\right)=F-Id, \mbox{ with } F=
\left(\begin{matrix}
  1&1&1&1\\
1&1&1&1\\
1&1&1&1\\
1&1&1&1
\end{matrix}\right)
$$

Since $\alpha_i=\phi_i^4$ all the $\phi_i$ must be fourth roots 
of unity and furthermore $\phi_1\phi_2\phi_3=1$.
There are again three cases. 

The first is that all $\phi_i=1$.
In this case $\bar\chi(\rho(A))=\bar\chi(\rho(B))=\bar\chi(\rho(C))=Id$.
Therefore $\bar\chi(\BTheta)=<\bar\chi(H)>\subset M_4(\C)$. Now $F^n=4^{n-1}F$
and so $\bar\chi(H)^n=Id+(\sum_{i=1}^n {n\choose i}4^{i-1})F$ so that
$\bar\chi(\BTheta)$ is the 2--dimensional subalgebra of $M_4(\C)$ 
spanned by $Id$ and $F$.

The second case is that  all $\phi_i^2=1$ and only one $\phi_i=1$.
We will assume that $\phi_1=1, \phi_2=-1,\phi_3=-1$. 
The rest of the cases are symmetric.

In this case the image of $\bar\chi$ is the 6--dimensional 
subalgebra of matrices
of the form
\begin{equation}
  \left(\begin{matrix}
    a&b&c&c\\
b&a&c&c\\
d&d&e&f\\
d&d&f&e  \end{matrix}\right)
\end{equation}
Let $\B'$ be the sub--algebra above. It is an exercise to check
that this is indeed a subalgebra.
 Then as $\bar\chi(\rho(A))=\bar\chi(\rho(B))=\bar\chi(\rho(C))=diag(1,1,-1,-1)$, so that $\bar\chi(\BTheta)\subset\B'$.
On the other hand $P_1=E_{11}+E_{22}=\frac{1}{2}(Id+diag(1,1,-1,-1))\in\BTheta$ and 
$P_2=E_{33}+E_{44}=\frac{1}{2}(Id-diag(1,1,-1,-1))\in \bar\chi(\BTheta)$. And hence all the $P_i\bar\chi(H)P_j,i,j=1,2 \in \bar\chi(\BTheta)$. These are the $2\times 2$ block matrices
making up $\bar\chi(H)$. But the span of these matrices is precisely $\B'$. 
 
The final case is exactly one $\phi_i^2=1$. We will treat the case
$\phi_1=1,\phi_2=i,\phi_3=-i$. The rest are again symmetric.
In this case $\bar\chi(\BTheta)$ is the 10--dimensional matrix algebra of matrices of the form

\begin{equation}
  \left(\begin{matrix}
    a&b&c&d\\
b&a&c&d\\
e&e&g&h\\
f&f&k&l  \end{matrix}\right)
\end{equation}
Let $\B''$ be the subalgebra above. Again it is an exercise
to check that $\B''$ is indeed a subalgebra. In the current case:
$\bar\chi(\rho(A))=\bar\chi(\rho(B))=\bar\chi(\rho(C))=diag(1,1,i,-i)$ 
and hence it follows that $\bar\chi(\BTheta)\subset \B''$.
Now $\bar\chi(\rho(A))^2=diag(1,1,-1,-1)$ so that the $P_1,P_2$ above
are in $\bar\chi(\BTheta)$. But furthermore we have that 
$\frac{-1}{2i}(\bar\chi(\rho(A))P_2\bar\chi(H)P_2-iP_2\bar\chi(H)P_2=E_{34}$ 
and 
$\frac{1}{2i}(\bar\chi(\rho(A))P_2\bar\chi(H)P_2=iP_2\bar\chi(H)P_2=E_{43}$. So that
$E_{33},E_{34},E_{43},E_{44}\in \bar\chi(\BTheta)$. The other generators of $\B'$
now are given by $P_1$, $P_1\bar\chi(H)P_1$, $P_2\bar\chi(H)P_1$,
$E_{44}P_2\bar\chi(H)P_1$, $P_1\bar\chi(H)P_2E_{33}$ and 
$P_1\bar\chi(H)P_2E_{33}$ which are all in $\bar\chi(\BTheta)$.

\end{proof}

The trickiest case is the case in which all the $\alpha_i$ are real and not all equal to one.

\begin{prop}
Let $\Phi=1$
and all the $\alpha_i$ real, but not all $\alpha_i=1$. Then if all the $\phi_i$ are different  $\BTheta=M_4(\T_{\Theta})$, else  $\BTheta\subsetneq M_4(\T_{\Theta})$.
\end{prop}
\begin{proof}
Since not all $\alpha_i=1$ there must be exactly one $\alpha_i=1$ with the other two being equal to $-1$. We will deal with the case $\alpha_3=1,\alpha_2=\alpha_3=-1$. The others are similar.
Consider $X=H+\rho(BC^*)H\rho(CB^*)$. $X^2=12E_{11}$ and $X_2=HE_{11}=E_{21}+E_{31}+E_{41}$.
We know that $\phi_1\neq \phi_2$ and $\phi_1\neq \phi_3$  and $X_3:=\phi_1X_2-\rho(A)X_2\rho(A^*)=
(\phi_1-\phi_2)E_{31}+(\phi_1-\phi_3)E_{41}$. Here we get the two cases.
If $\phi_2\neq \phi_3$ then we obtain $\phi_2X_3-\rho(A)X_3\rho(A^*)=(\phi_2-\phi_3)(\phi_1-\phi_3)E_{31}$ and we have successively $E_{31}, E_{41},E_{21}$ and their transposes in $\BTheta$.
But this is a set of generators.

In case $\phi_2=\phi_3$, we see that $P_1=\frac{1}{\phi_1-\phi_2}X_2=E_{31}+E_{41}$. Hence also $E_{12},E_{21},E_{22}\in \BTheta$. Furthermore $P_2=P_1E_{12}=E_{32}+E_{42}\in\BTheta$ and $Q_1=Id-E_{11}-E_{22}=E_{33}+E_{44} \in \BTheta, Q_2=Q_1HQ_1$.

We shall use the morphism $\chi:\T^3_{\Theta}\to Cl:=\Cliff(\left(\begin{matrix}1&0\\0&1\end{matrix}\right))\otimes \C$ given by $\chi(A)=\chi(B)=e_1$,
 $\chi(C)=e_2$ where $e_i, i=1,2$ are the generators of the Clifford algebra which satisfy $e_i^2=1$
 and $e_1e_2+e_2e_1=1$. The map $\chi$ induces an algebra morphism $\bar\chi:M_4(\T^3_{\Theta})\to Cl$.
 
 {\bf Claim:} The image of $\bar\chi$ is given  matrices of the form
\begin{equation}
  \left(\begin{matrix}
    a&b&e&e\\
c&d&f&f\\
g&h&k&l\\
g&h&l&k  \end{matrix}\right)
\end{equation}
which is a free rank 10 $Cl$ proper submodule of $M_4(Cl)$ and in turn $\BTheta$ is a proper submodule
of $M_4(\T^3_{\Theta})$.

To prove the Claim, we again check both inclusions. Let $\B'$ be the subalgebra above.
Notice that $Id\in \B'$
$$
\bar\chi(H)=\left(\begin{matrix}0&1&1&1\\ 1&0& e_1&e_1\\1&e_1&0&e_2\\
1&e_1&e_2&0\end{matrix}\right)\in \B'
$$
 $\bar\chi\rho(A)=\bar\chi\rho(B)=e_1diag(1,\phi_1,\phi_2,\phi_2)\in \B'$ 
and $\bar\chi\rho(C)=e_2diag(1,\phi_1,\phi_2,\phi_2)\in \B'$. So that $\bar\chi(\BTheta)\subset \B'$.
On the other hand the algebra $\B'$ is generated by $E_{11},E_{12},E_{21},E_{22},
P_1,P_1^T,P_2,P_2^T,Q_1,e_2\bar\chi(Q_2)=E_{34}+E_{43}$ which are all in $\bar\chi(\BTheta)$,
so that $\BTheta\subset \bar\chi(\BTheta)$. 
\end{proof}

\begin{thm}
If $\Phi\neq1$ or $\Phi=1$ and at least one $\alpha_i\neq 1$ and all $\phi_i$ are different then
$\BTheta=M_4(\T^3_{\Theta})$ and $K(\BTheta)=\Z^4$.

If $\phi_i=1$ for all $i$ then $K(\BTheta)=K(X)$ where $X$ 
is the cover of the three torus given by Theorem \ref{commthm}.
\end{thm}

\begin{proof}
The only thing that remains to be proved is that $H$ is generic.
Indeed this can be done by direct computation. We will not
give the Eigenvalues here as they are quite long expressions. 
But notice that
for $\chi(A)=-1, \chi(B)=1,\chi(C)=-1$ the matrix $\bar\chi(H)$ has
the four distinct Eigenvalues $\pm\sqrt{5},\pm 1$.
\end{proof}

\section*{Acknowledgments}
We would like to thank J.~Bellissard, M.~Dadarlat, M.~Kontsevich, M.~Marcolli
and A.~Schwarz for enlightening discussions. 
Foremost, we wish to thank H.W.~Hillhouse for sparking the interest in this project and
many discussions. 

RK thankfully acknowledges 
support from NSF DMS-0805881 and the Humboldt Foundation. 
He also thanks the
Institut des Hautes Etudes Scientifiques, 
the Max--Planck--Institute for Mathematics in Bonn and the
 Institute for Advanced Study for their support 
and the University of Hamburg for its hospitality. 
While at the IAS RK's work was supported by the NSF under agreement  
DMS--0635607.

BK  thanks the Department of Physics at Princeton University for its hospitality and thankfully acknowledges support from the  NSF under the grant PHY-0969689.

  Any opinions, findings and conclusions or 
recommendations expressed in this
 material are those of the authors and do not necessarily 
reflect the views of the National Science Foundation.

\end{document}